\documentclass[numbers]{article}

\usepackage[preprint]{neurips_2024}

\usepackage[utf8]{inputenc} % allow utf-8 input
\usepackage[T1]{fontenc}    % use 8-bit T1 fonts
\usepackage{hyperref}       % hyperlinks
\usepackage{url}            % simple URL typesetting
\usepackage{booktabs}       % professional-quality tables
\usepackage{amsfonts}       % blackboard math symbols
\usepackage{nicefrac}       % compact symbols for 1/2, etc.
\usepackage{microtype}      % microtypography
\usepackage{xcolor}         % colors
\usepackage[pdftex]{graphicx}

\usepackage{booktabs}  
\usepackage{amsmath,bm, amsthm}
\usepackage{makecell}
\usepackage{thmtools, thm-restate}
\usepackage{arydshln}
\usepackage{amssymb}
\usepackage{algorithm}
\usepackage{algorithmic}
\usepackage{xcolor}
\usepackage{wrapfig}

\newtheorem{lemma}{Lemma}
\newtheorem{theorem}{Theorem}
\newtheorem{definition}{Definition}

\title{Gaussian is All You Need: A Unified Framework for\\ Solving Inverse Problems via Diffusion Posterior Sampling}

\author{Nebiyou Yismaw  \\
University of California Riverside\\
\texttt{nyism001@ucr.edu} \\
\And
Ulugbek S. Kamilov  \\
Washington University in St. Louis \\
\texttt{kamilov@wustl.edu}
\AND
M. Salman Asif   \\
University of California Riverside \\
\texttt{sasif@ucr.edu} \\
}
\begin{document}

\makeatletter
\renewcommand{\@noticestring}{}
\makeatother
\maketitle

\begin{abstract}
Diffusion models can generate a variety of high-quality images by modeling complex data distributions. Trained diffusion models can also be very effective image priors for solving inverse problems. Most of the existing diffusion-based methods integrate data consistency steps by approximating the likelihood function within the diffusion reverse sampling process. 
In this paper, we show that the existing approximations are either insufficient or computationally inefficient. 
To address these issues, we propose a unified likelihood approximation method that incorporates a covariance correction term to enhance the performance and avoids propagating gradients through the diffusion model. The correction term, when integrated into the reverse diffusion sampling process, achieves better convergence towards the true data posterior for selected distributions and improves performance on real-world natural image datasets. Furthermore, we present an efficient way to factorize and invert the covariance matrix of the likelihood function for several inverse problems. Our comprehensive experiments demonstrate the effectiveness of our method over several existing approaches. Code available at \href{https://github.com/CSIPlab/CoDPS}{https://github.com/CSIPlab/CoDPS}.

\end{abstract}

\section{Introduction}

Diffusion-based models \cite{ho2020denoising,sohl2015deep, song2019generative, song2020score} have recently gained attention due to their powerful generative ability by learning complex data distributions. A number of recent methods have used diffusion-based models as priors for solving inverse problems \cite{kawar2022denoising, zhu2023denoising}. 
Diffusion Posterior Sampling (DPS) \cite{chung2023diffusion} and $\Pi$GDM \cite{song2022pseudoinverse} are notable examples that incorporate widely adopted posterior sampling schemes in the reverse diffusion process. Nevertheless, both methods are computationally inefficient as they require computing gradients through the diffusion model to compute the conditional score. Furthermore, DPS uses inexact likelihood covariance matrix estimates (even for simple Gaussian priors).

\begin{figure}[t]
    \centering
    \includegraphics[width=\columnwidth]{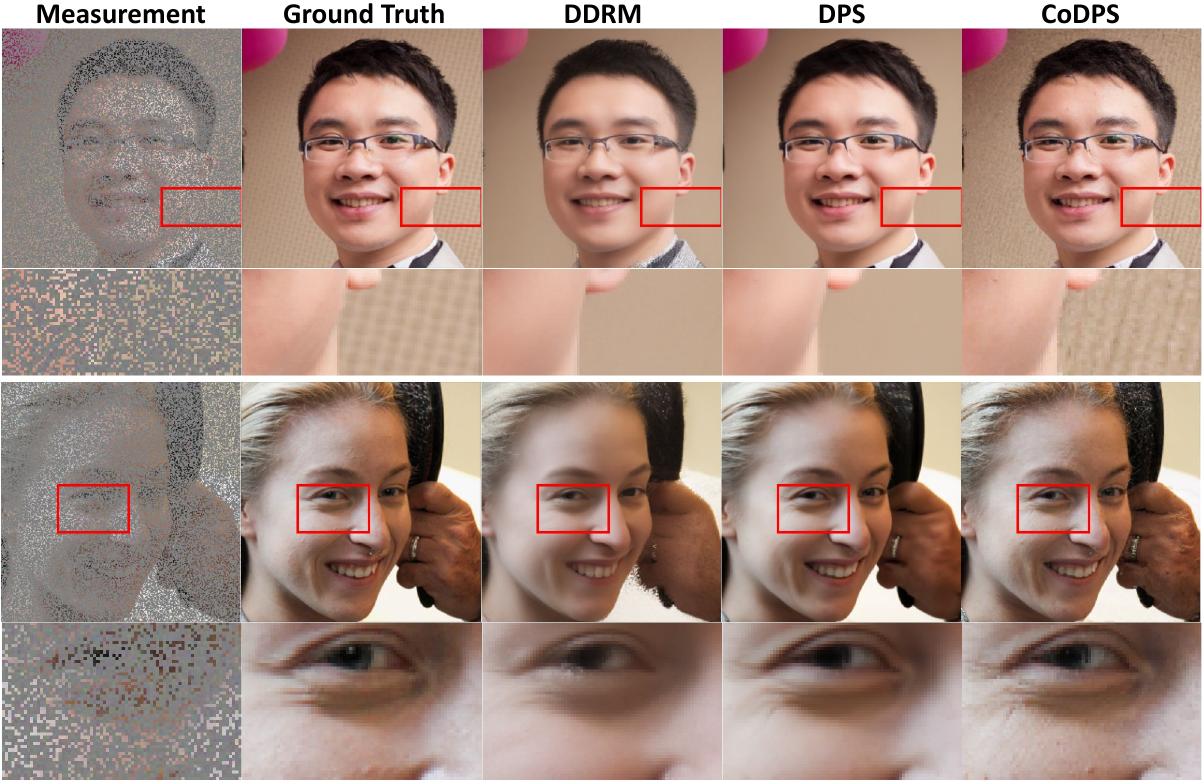}
    \caption{
    Our method effectively recovers fine details, as clearly shown in the zoomed-in images, resulting in outputs that are more consistent with the ground truth.  Particularly, the low-level details, such as the background pattern in the first image and the facial texture in the second, are visible only in the outputs produced by our method, whereas other approaches fail to restore these details.
    }
    \label{fig:ffhq_inp}
\end{figure}

\textbf{In this paper, we show that the Gaussian prior assumption is all that is needed to achieve the best of both worlds, offering both covariance correction and computational efficiency.} Building on this assumption, we propose a novel and simple sampling framework that is not only computationally efficient but also achieves performance that is competitive with, and in some cases surpasses, existing methods. 
Our work begins by computing the distribution of measurements conditioned on intermediate outputs of the diffusion network for simple Gaussian priors. For a linear inverse problem with a Gaussian prior, the conditional distribution has an explicit Gaussian form. This distribution has the same form as the conditional distribution proposed in $\Pi$GDM, even though both methods start from different assumptions (see Sec.~\ref{subsec:prior_approx}). Furthermore, for the isotropic covariance case, our conditional distribution is equivalent to the one used in DPS, up to a scalar constant (see Sec.~\ref{subsec:prior_approx}). While the conditional distribution is generally intractable for non-Gaussian priors, the form obtained under the Gaussian assumption serves as an approximation that results in remarkable reconstruction performance. 
Figure~\ref{fig:ffhq_inp} presents some examples that illustrate how our method successfully recovers fine details that other methods fail to restore. % Additional results are presented in Sec.~\ref{sec:exp?}. 
\textbf{Another key advantage of this assumption is that it allows us to bypass gradient computation through the network when computing the score of the conditional distribution.} Our experiments show that bypassing gradients through the network can  provide at least $2\times$ reduction in the computational cost. 

{The Gaussian assumption for approximating the conditional score is already used, both implicitly and explicitly, in different diffusion-based posterior sampling methods}. For instance, in $\Pi$GDM  the conditional Gaussian assumption can be shown to be equivalent to the Gaussian prior assumption. As a consequence of this assumption, the method was able to introduce a covariance term in the likelihood. 
A recent method, DDS \cite{chung2024decomposed}, introduces a novel sampling strategy that entirely bypasses gradient computations through the network. Instead, it approximates the gradient by projecting orthogonally onto the clean data manifold. This approach significantly simplifies the sampling process by assuming that the prior distribution is Gaussian with infinite variance.

In this paper, we introduce a \textbf{Co}variance Corrected \textbf{D}iffusion \textbf{P}osterior \textbf{S}ampling (CoDPS) method, which utilizes a likelihood approximation method with a covariance correction term that depends on the diffusion time steps and the forward model. This approximation is exact for Gaussian priors, and we demonstrate that it can improve reconstruction performance in real-world restoration problems. 
% As illustrated in Figure~\ref{fig:ffhq_inp}, our method excels at recovering fine details that are often lost by other approaches. It is able to restore low-level features, as shown in the zoomed views, with outputs closely matching the ground truth.
Since the correction term is time-dependent, we can encode the uncertainty in our conditional guidance. Specifically, at the early stages of the diffusion process, the uncertainty in our estimates is high and gradually decreases as $t \to 0$. 

{We summarize our main contributions as follows.}
\begin{itemize}
    \item We first demonstrate that by making a simple choice of a Gaussian prior, we get a conditional score that is accurate under this prior and eliminates the need for expensive gradient computations through the diffusion network. 
    \item Building on the Gaussian score approximation, we propose Covariance Corrected Diffusion Posterior Sampling (CoDPS), a novel framework for solving inverse problems. This framework significantly reduces time and memory requirements.
    \item We further accelerate CoDPS by efficient factorization and inversion of the high-dimensional covariance matrices for a family of inverse problems.
    \item We validate the accuracy of CoDPS through a proof-of-concept experiment on a Gaussian mixture and demonstrate its effectiveness on real-world image datasets. Our experiments show that CoDPS  achieves competitive performance with remarkable computational efficiency. 
\end{itemize}
% -------------------------------------------------------------
\section{Background}

\noindent\textbf{Notations.} 
In this paper, we denote scalars, vectors, matrices, and operators (or functions) by $x$, $\bm{x}$, $\bm{X}$, and $\mathcal{X}(\cdot)$, respectively. 
We use $\bm{X}^\mathsf{H}, \mathcal{X}^\mathsf{H}$ to denote the Hermitian transpose (i.e., adjoint operation) on matrices and operators. 
To simplify the notations, we represent 2D  images as 1D vectors, but all the image transformations are performed in 2D. 
We will use $\bm{x}_t$ or $\bm{x}(t)$ to denote time-dependent variables. We use $\sigma^2_i, \bm{\Sigma_{i}}$ and $\sigma^2_{i|j},\bm{\Sigma_{i|j}}$ to represent % $\bm{\Sigma_{\bm{x}_i|\bm{x}_j}}$, which is 
the variance, covariance for $p(\bm{x}_i)$ and $p(\bm{x}_i|\bm{x}_j)$, respectively. 

\subsection{Diffusion models}

The goal of DDPMs \cite{ho2020denoising, sohl2015deep} is to model a complex data distribution $q(\bm{x}_0)$ through a forward Markov chain that incrementally transforms structured data into pure noise, and a reverse process that reconstructs data from the noise. With a given noise schedule  $\beta_t \in (0,1)$ for  $t = 1, \dots, T$, the forward process creates noisy perturbations using a Gaussian transition kernel. 
We can express the forward transition to any time $t$ as
\begin{align}
\label{eq:fwd_q_xt_x0}
q(\bm{x}_t|\bm{x}_0) = \mathcal{N}(\bm{x}_t; \sqrt{\bar{\alpha}_t} \bm{x}_0, (1 - \bar{\alpha}_t) \bm{I}),
\end{align}
where $\alpha_t = 1 - \beta_t$ and $\bar{\alpha}_t = \prod_{s=1}^t \alpha_s$. The reverse Markov process has a similar functional form as the forward process and uses learned Gaussian transition kernels $p_\theta(\bm{x}_{0:T})$. The learned reverse diffusion process is modeled as
\begin{align}
\label{eq:p_xtm1_xt}
p_\theta(\bm{x}_{t-1}|\bm{x}_t) = \mathcal{N}(\bm{x}_{t-1}; \bm{\mu}_\theta(\bm{x}_t, t), \bm{\Sigma}_\theta(\bm{x}_t, t)).
\end{align}
During training, we aim to learn $\bm{\mu}_\theta$ and $\bm{\Sigma}_\theta$, which are used to predict $\bm{x}_{t-1}$ given $\bm{x}_t$ during the reverse sampling process. 
% This corresponds to learning a denoiser at different time steps $t$, where the standard deviation of the noise added is a function of the time step. 
The training objective of a diffusion model minimizes the KL divergence between the joint distributions $q(\bm{x}_{0}, \dots, \bm{x}_T)$ and $p(\bm{x}_0, \dots, \bm{x}_T)$. 
The parameterization in \cite{ho2020denoising} demonstrates that the trained denoiser $\bm{\epsilon}_\theta$ can be used to predict the mean,  {which we denote as $\hat{\bm{x}}_0$}: 
\begin{equation}
\label{eq:mean_p_x0_xt}
{\hat{\bm{x}}_0} = \frac{1}{\sqrt{\bar{\alpha}_t}} \left( \bm{x}_t - \sqrt{1 - \bar{\alpha}_t} \bm{\epsilon}_\theta(\bm{x}_t) \right).
\end{equation}
Similarly, the covariance $\bm{\Sigma}_\theta(\bm{x}_t, t)$ can be learned or set as a fixed parameter, such as $\sigma_t^2 \bm{I}$. 
Once the mean predictor $\bm{\epsilon}_\theta$ and the covariance $\sigma_t^2 \bm{I}$ are trained or determined, they can be used in the reverse sampling process. 

% DDPMs can be considered a discretized formulation of the more general 
Score SDE models \cite{song2020score}
% The perturbation and reverse processes of Score SDEs are solutions to stochastic differential equations. 
% The goal of these methods is to 
learn the score $\nabla_{\bm{x}_t} \log p(\bm{x}_t)$ with a diffusion network $\bm{s}_\theta(\bm{x}_t,t)$ using denoising score matching objectives \cite{raphan2006learning, raphan2011least, vincent2011connection}.
Once the score network is trained, we can use it to obtain samples from regions of high probability. 
The score function can be approximated by a denoising model that learns the underlying probability distribution of the training data, as shown in previous work \cite{kadkhodaie2021stochastic, song2020score}. The relationship between the denoiser network $\bm{\epsilon}_\theta(\bm{x}_t)$ and the score function $\bm{s}_\theta(\bm{x}_t, t)$ is expressed as
\begin{equation}
\bm{s}_\theta(\bm{x}_t, t) = \frac{-\bm{\epsilon}_\theta(\bm{x}_t)}{\sqrt{1-\bar{\alpha}_t}}.
\end{equation}
This equation highlights how the output of the denoiser is linked to the gradient of the log likelihood of the observed data. By establishing this connection, we can view \eqref{eq:mean_p_x0_xt} as an unconditional estimation informed by the prior implicit in the denoiser. In Section \ref{sec:diff_inv_prob}, we show how this unconditional sample generation can be utilized to solve inverse problems, where we aim to obtain samples from high-probability regions that are also consistent with our measurements.

\subsection{Posterior sampling for inverse problems}
\label{sec:diff_inv_prob}
% \textit{Discuss diffusion model used for inverse problems, DDRM, DPS, DiffPIR and other related works}
An inverse problem is the recovery of an unknown signal $\bm{x}_0 \sim p(\bm{x}_0)$ from a set of measurements:
\begin{equation}
    \bm{y} = \mathcal{A}(\bm{x}_0) + \bm{\eta}, \quad \bm{\eta} \sim \mathcal{N}(\bm{0}, \sigma_n^2 \bm{I}),
    \label{eq:fwd_model}
\end{equation}
where $\bm{y}$ represents measurements,  $\mathcal{A}(\cdot)$ represents a forward operator or a measurement matrix, and $\sigma_n$ represents noise. The maximum a posteriori (MAP) ${\bm{x}}_\text{MAP}$ estimate can be found by maximizing the sum of the log likelihood $\log p(\bm{y}|\bm{x}_0)$ and the log of the data prior distribution $\log p(\bm{x}_0)$. From equation \eqref{eq:fwd_model}, it is evident that $p(\bm{y}|\bm{x}_0)$ is a Gaussian distribution. However, identifying a suitable prior $p(\bm{x}_0)$ remains an active area of research.
Pre-trained Diffusion models are suitable candidates as they learn the log prior density $\nabla_{\bm{x}} \log p(\bm{x})$. 
{To utilize them for solving inverse problems, we build upon the unconditional denoising prediction in \eqref{eq:mean_p_x0_xt} to develop a measurement-conditioned sampling procedure.
} This involves incorporating a conditional score in the reverse sampling process as
\begin{align}
\label{eq:map_diffusion}
    \nabla_{\bm{x}_t} \log p(\bm{x}_t|\bm{y}) = \nabla_{\bm{x}_t}\left(\log p(\bm{x}_t) + \log p(\bm{y}|\bm{x}_t)\right).
\end{align}
% In the right-hand side of \eqref{eq:map_diffusion}, there are two terms that need modeling.
The term $\nabla_{\bm{x}_t} \log p(\bm{x}_t)$ in \eqref{eq:map_diffusion} can be approximated using a pre-trained diffusion model $\bm{s}_\theta (\bm{x}_t,t)$. However, as noted in \cite{song2022pseudoinverse, chung2023diffusion}, $p(\bm{y}|\bm{x}_t)$ is intractable. This issue is discussed in DPS \cite{chung2023diffusion} and $\Pi$GDM \cite{song2022pseudoinverse}, where it is highlighted that there is no analytical form available for $p(\bm{y}|\bm{x}_t)$. 
%%% This might be redundant, citation is enough
To understand why this is the case, let us consider
\begin{equation}
\label{eq:marg_y_xt}    
p(\bm{y}|\bm{x}_t) = \int p(\bm{y}|\bm{x}_0,\bm{x}_t) p(\bm{x}_0|\bm{x}_t) \, d\bm{x}_0.
\end{equation}

Since $\bm{x}_t$ is obtained from $\bm{x}_0$ by corrupting it with noise, once we know $\bm{x}_0$, knowing $\bm{x}_t$ does not provide any additional information about $\bm{y}$. For this reason, we can write the conditional as
$
p(\bm{y}|\bm{x}_t) = \int_{\bm{x}_0} p(\bm{y}|\bm{x}_0) p(\bm{x}_0|\bm{x}_t) \, d\bm{x}_0.
$
The main challenge then lies in modeling $p(\bm{x}_0|\bm{x}_t)$. Using the reverse diffusion sampling process, we can approximate and sample from this distribution. However, marginalizing over $\bm{x}_0$ is infeasible, as generating even a single sample requires performing a complete reverse ancestral sampling. 
To circumvent this, DPS\cite{chung2023diffusion} proposed to approximate the conditional score $\bm{\kappa}_t$ as
\begin{align}
    \label{eq:dps_approx}
    \bm{\kappa}_t = \nabla_{\bm{x}_t}\log p(\bm{y}|\bm{x}_t) \simeq 
    \nabla_{\bm{x}_t}\log p(\bm{y}|\hat{\bm{x}}_0),
\end{align} 
where $\hat{\bm{x}}_0$ is given by \eqref{eq:mean_p_x0_xt} and $ p(\bm{y}|\hat{\bm{x}}_0) \simeq \mathcal{N}(\mathcal{A}(\hat{\bm{x}}_0), \sigma^2_n \bm{I} )$. 
Alternatively, $\Pi$GDM \cite{song2022pseudoinverse} assumes $p(\bm{x}_0|\bm{x}_t) \sim \mathcal{N}(\hat{\bm{x}}_0, r_t^2 \bm{I})$, where $r_t$ is a time- and data-dependent hyper-parameter, and approximates the conditional score as
\begin{align}
    \label{eq:pigdm_approx}
    \bm{\kappa}_t = \nabla_{\bm{x}_t}\log p(\bm{y}|\bm{x}_t) \simeq 
    \left ( 
    (\bm{y} - \bm{A}\hat{\bm{x}}_0)^\mathsf{H}\bm{\Sigma}_t^{-1} \bm{A} \dfrac{\partial \hat{\bm{x}}_0}{\partial \bm{x}_t}\right)^\mathsf{H},
\end{align}
where $\bm{\Sigma}_t=r_t^2 \bm{A}\bm{A}^\mathsf{H} + \sigma_n^2 \bm{I}$ . Both approximations in DPS and $\Pi$GDM have demonstrated remarkable results on various real-world datasets. However, the approximation in DPS may become inaccurate as $t$ approaches the total number of diffusion steps $T$. Additionally, both methods compute backpropagation through the diffusion network to evaluate $\partial \hat{\bm{x}}_0/\partial \bm{x}_t$, which can significantly increase the time and memory complexity of the sampling process. 

Several recent works have proposed strategies to eliminate the expensive gradient computation through the diffusion network during inverse problem solving. One such approach is DMPS~\cite{meng2022diffusion}, which introduces a closed-form approximation of the conditional likelihood by assuming an uninformative prior on $p(\bm{x}_0)$. This assumption allows the method to avoid backpropagation entirely. Notably, our Gaussian prior also represents DMPS as a special case when $\sigma_0^2 \to \infty$.
DDNM~\cite{wangzero} introduces a zero-shot framework for solving linear inverse problems by leveraging the structure of the forward operator. It decomposes the solution into two orthogonal components: the row space and the null space of the operator $\bm{A}$. The row-space component is obtained directly via the pseudoinverse projection $\bm{A}^\dagger \bm{y}$, ensuring strict data consistency. The null-space component is refined by iteratively denoising with a pretrained diffusion model and projecting the result onto the null space using $(\mathbf{I} - \bm{A}^\dagger \bm{A})\hat{\bm{x}}_0$. However, this method is unstable in the presence of measurement noise. To address this, the authors proposed DDNM+, which includes a range-space correction term. In our experiments, we observed poor performance of DDNM+ on deblurring tasks in the presence of noise, despite the correction term.
ReSample~\cite{songsolving} is another method that leverages pretrained latent diffusion models (LDMs). It avoids gradient computation by enforcing hard data consistency in the latent space through constrained optimization, a process referred to as a ReSample timestep. 
ReSample applies unconditional DDIM updates in the early stages, and resampling is performed every $k$ steps (typically $k=10$).
Despite only applying sampling at certain steps, ReSample has significantly higher runtime. Additionally, it requires both an encoder and a decoder to map data to and from the latent space, which adds complexity and limits its applicability in settings without pretrained auto-encoder models. 

DiffPIR~\cite{zhu2023denoising} follows a plug-and-play strategy and uses pretrained diffusion models as generative denoiser priors. While DiffPIR does not explicitly model the conditional posterior, the sampling steps in DiffPIR closely resemble the steps in DPS and CoDPS (albeit in a different order, as a result of  different assumptions). Our experimental results show that CoDPS provides better and improved results compared to DPS and DiffPIR. Ather representative method is DDRM~\cite{kawar2022denoising}, which performs spectral-domain posterior sampling but relies on SVD of the forward operator and assumes linearity. DiffPIR~\cite{zhu2023denoising} follows a plug-and-play strategy, applying pretrained diffusion models as proximal denoisers without explicitly modeling the posterior. Recent extensions~\cite{peng2024improving} propose covariance estimation modules to enhance conditioning accuracy in such frameworks.

A complementary line of work, referred to as variational inference-based methods, approximates the posterior distribution with a simpler, tractable distribution instead of seeking for a closed-form expression for the conditional score ~\cite{daras2024survey}. RED-Diff~\cite{mardanivariational} proposes to minimize the KL divergence between the true conditional posterior $p(\bm{x}_0 \mid \bm{y})$ and a variational distribution $q(\bm{x}_0 | \bm{y}) = \mathcal{N}(\bm{\mu}, \sigma^2 \bm{I})$. The mean $\bm{\mu}$ is iteratively updated by minimizing the KL objective, while the variance of the variational distribution is assumed to be small $\sigma^2 \approx 0$. While RED-Diff and CoDPS share the goal of modeling the posterior that utilize simple Gaussian forms, the two approaches differ significantly in formulation and implementation. RED-Diff formulates  posterior sampling as a stochastic optimization problem that requires iterative optimization with reverse diffusion sampling in each step. In contrast, CoDPS derives a closed-form correction to the score under a Gaussian prior without any need to perform an iterative optimization problem.

Other approaches, such as \cite{chung2022improving, song2020score, choi2021ilvr, lugmayr2022repaint, murata2023gibbsddrm, kawar2021snips}, utilize pre-trained diffusion models. In this work, we demonstrate the connections between DPS, $\Pi$GDM, and DDS-type conditioning methods, and propose a unified diffusion posterior sampling algorithm. Our approach leverages implicit assumptions present in previous works and enhances computational efficiency by eliminating gradient computation through the network.
While the focus of our work is mainly on unsupervised diffusion based methods, there are several deep learning-based solvers trained in a fully supervised manner \cite{saharia2022image, rombach2022high, whang2022deblurring, xia2023diffir, delbracio2023inversion, yismaw2024domain}.

\section{Method}

We begin this section by discussing the limitations of existing methods for approximating the conditional score for $p(\bm{y}|\bm{x}_t)$. We highlight the differences and explore connections between the implicit assumptions made by these methods. Next, we present our proposed \textbf{Co}variance \textbf{Co}rrected \textbf{D}iffusion \textbf{P}osterior \textbf{S}ampling (\textbf{CoDPS}) method and explain how it can be incorporated into reverse diffusion sampling algorithms such as DDPM \cite{ho2020denoising} and DDIM \cite{song2020denoising}. We describe how our method builds upon assumptions that are already present in previous methods. Finally, we present efficient covariance inversion methods for a broad family of inverse problems.

\begin{figure*}[th]
    \centering
    \includegraphics[width=\textwidth]{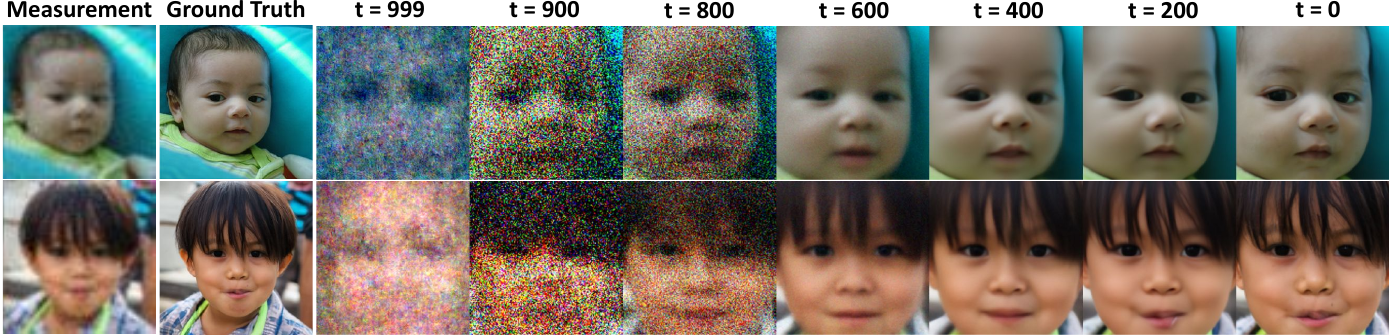}
    \caption{Evolution of $\hat{\bm{x}}_0$ as ${t \to 0}$ for solving a noisy super-resolution problem. For large values of $t$, $\hat{\bm{x}}_0$ appears noisy and does not resemble a natural image. For intermediate values of $t$, the images appear plausible but may diverge from the ground truth in terms of pose, appearance, and other visual details. This highlights that the uncertainty in $\hat{\bm{x}}_0$ is initially high and decreases over time, an effect that should be considered in conditional score estimation.
    }
    \label{fig:X_0_hat_evolution}
\end{figure*}

\subsection{Conditional score approximation}
\label{subsec:prior_approx}
We revisit the conditional score approximation from DPS \cite{chung2023diffusion}, stated as $p(\bm{y}|\bm{x}_t) \approx \mathcal{N}(\bm{A}\hat{\bm{x}}_0, \sigma_n^2 \bm{I} )$ for a linear inverse problem under additive Gaussian measurement noise, where $\hat{\bm{x}}_0$ is the quantity given in \eqref{eq:mean_p_x0_xt}. Intuitively, this approximation appears to be valid for small values of $t$ as the diffusion model learns the true data distribution $p(\bm{x}_0)$ and the clean image estimates $\hat{\bm{x}}_0$ will accurately represent samples from this distribution. However, at the early stages of the reverse diffusion process $t \approx T$, the samples $\bm{x}_t$ will be close to an isotropic Gaussian distribution. The clean image estimates $\hat{\bm{x}}_0$ from $\bm{x}_t$ at these stages will not necessarily come from the true prior distribution. This can be easily verified by looking at the estimates $\hat{\bm{x}}_0$ for different values of $t$, as shown in Figure \ref{fig:X_0_hat_evolution}. When $ 600 < t < 1000 $, the estimates do not seem to resemble clean natural images (i.e samples from $p(\bm{x}_0)$). Additionally, within the range $ t \in [200,600] $, although the samples appear plausible, they are perceptually different from the ground truth image.
More importantly, the uncertainty in the measurement $\bm{y}$ when conditioned on these estimates grows as $t$ increases and needs to be accounted for in the conditional variance. To see this more clearly, let us consider a simple prior Gaussian distribution given by $\bm{x}_0 \sim \mathcal{N}(0, \sigma_0^2 \bm{I})$. Using the properties of marginal and conditional Gaussian distributions, we obtain 
\begin{align}
    \label{eq:p_x0_xt_approx}
    p(\bm{x}_0|\bm{x}_t) &= \mathcal{N}\left(\bm{\Sigma}_{0|t}\left(\dfrac{\sqrt{\bar\alpha_t}}{1-\bar\alpha_t} \bm{x}_t \right), \bm{\Sigma}_{0|t}\right),
\end{align}
where
\begin{align}
\label{eq:var_x0_xt_approx}
\bm{\Sigma}_{0|t} = \dfrac{\sigma_0^{2}(1-\bar\alpha_t)}{(1-\bar\alpha_t) + \sigma_0^{2}\bar\alpha_t} \bm{I} = \sigma^2_{0|t}\bm{I}.
\end{align}
As $t \to 0$, $\bar{\alpha}_t \to 1 $  and $p(\bm{x}_0|\bm{x}_t) \to  \delta\left(\bm{x}_0 - \frac{1}{\sqrt{\bar{\alpha}_t}}\bm{x}_t\right)
$ (i.e., a delta distribution at the scaled noisy estimate $\bm{x}_t$). Similarly, as $t \to T$, where $T$ is the number of diffusion steps, $\bar{\alpha}_t \to 0 $ and $p(\bm{x}_0|\bm{x}_t) \to \mathcal{N}(0, \sigma_0^2 \bm{I})$, which is the prior data distribution. Note that the posterior mean provided by \eqref{eq:p_x0_xt_approx} matches the posterior mean obtained through Tweedie’s identity in \eqref{eq:mean_p_x0_xt} (see more details in  Appendix \ref{subsec:posterior_mean_adx}). 

In contrast, $\Pi$GDM \cite{song2022pseudoinverse} assumes $p(\bm{x}_0|\bm{x}_t) \sim \mathcal{N}(\hat{\bm{x}}_0, r_t^2 \bm{I})$, which leads to the approximation in \eqref{eq:pigdm_approx}. This approximation is equivalent to DPS up to a scalar constant for cases where $\sigma_n^2 \gg r_t$, or for the cases where the forward operator has orthogonal rows ( i.e $\bm{A}\bm{A^\mathsf{H}} = \bm{I}$). Furthermore, by applying Theorem \ref{theorem:cond_gauss}, we can draw a direct connection between the underlying assumptions made by $\Pi$GDM and DDS. In particular, we show that the conditional assumption in $\Pi$GDM is equivalent to the Gaussian assumption in DDS.

\begin{restatable}[Conditionally specified Gaussian distributions \cite{arnold1999conditional}]{theorem}{condgauss}
\label{theorem:cond_gauss}
. Let $X$ and $Y$ be random variables, and let $p(X|Y)$ and $p(Y|X)$ be both Gaussian distributions. Then, the joint distribution $f_{X,Y}(x,y)$ is a bivariate normal distribution if either $var(Y|X)$ or $var(X|Y)$ is constant.
\end{restatable}
\noindent
We provide the proof of Theorem \ref{theorem:cond_gauss} in Appendix \ref{subsec:proof_theo_1}. The forward diffusion in \eqref{eq:fwd_q_xt_x0} ensures that $p(\bm{x}_t| \bm{x_0})$ is a normal distribution. We note that for a given $\bm{x}_t$, $r_t$ in \eqref{eq:pigdm_approx} is a constant. Theorem~\ref{theorem:cond_gauss} further shows that the joint distribution $p(\bm{x}_0, \bm{x}_t)$ is a bivariate normal distribution. 
From this we can easily deduce that the marginals $p(\bm{x}_0)$ and $p(\bm{x}_t)$ are also Gaussian. Following the arguments above, we can conclude that the Gaussian assumption on $p(\bm{x}_0|\bm{x}_t)$ implies that $p(\bm{x}_0)$ itself is a Gaussian. 

In this paper, we demonstrate that assuming the specific Gaussian distribution for $p(\bm{x}_0)$ allows for incorporating a covariance correction into the conditional likelihood score. It also enables the approximation of the Jacobian, as shown in DDS \cite{chung2024decomposed}. These techniques yield a unified posterior sampling method that enhances reconstruction performance while improving time and memory efficiency.

\subsection{Covariance corrected
posterior sampling}

Our proposed posterior sampling process, described in Algorithm~\ref{alg:co_dps_ddim}, involves iterative steps to enforce prior and measurement likelihood consistency. The prior consistency step uses a diffusion model $\bm{s}_\theta$, pre-trained to learn the score function (i.e., the gradient of the log prior), to obtain estimates $\bm{x}_t$, as outlined in lines 3-5. Specifically, we apply the DDIM sampling scheme \cite{song2020denoising} to predict the denoised estimate $\bm{\hat{x}}_0$ in line 5. Based on the time index $i$, we extract time-step $t$ and noise-schedule values $\alpha_t$, which will be used to compute the DDIM coefficients $c_1$ and $c_2$. These coefficients are then used to compute the estimates for the subsequent iterations in lines 6 and 7. 
{Finally, we update our latest estimate by moving in the direction of the gradient of the log of the measurement likelihood (line 8). However, we note that computing the gradient with respect to $\bm{x}_t$ is computationally expensive because it requires backpropagation through the diffusion model. }

\begin{algorithm}[t]
    \caption{Covariance Corrected
Diffusion Posterior Sampling (CoDPS)}
   \label{alg:co_dps_ddim}
    \begin{algorithmic}[1]
     \REQUIRE $N$, $\bm{y}$, \{$\zeta_t\}_{t=1}^N,  {\{\tilde\sigma_t\}_{t=1}^N}$, $\sigma_0^2$
     \STATE $\bm{x}_N \sim \mathcal{N}(\bm{0}, \bm{I})$
      \FOR{$i=N-1$ {\bfseries to} $0$}
       
           \STATE{$t, \bar{\alpha}_t, \bar{\alpha}_{t-1} \gets \text{extract\_alpha\_t}(i)$ \hfill \textit{$\triangleright$ Get alpha values}}
    
         \STATE{{$\hat{\bm{s}} \gets \bm{s}_\theta(\bm{x}_t, t)$}}
         
         \STATE{{$\hat{\bm{x}}_0 \gets 
         \dfrac{1}{\sqrt{\bar\alpha_t}}(\bm{x}_t + (1 - \bar\alpha_t)\hat{\bm{s}})$}}
         
         \STATE{$\bm{z} \sim \mathcal{N}(\bm{0}, \bm{I})$}
         
         \STATE{$\bar{\bm{x}}_{t} \gets \sqrt{\bar{\alpha}_{t-1}}\hat{\bm{x}}_0 + c_1\bm{z} + c_2 \hat{\bm{s}} $ \hfill \textit{$\triangleright$ DDIM coefficients}}
         \STATE{
         $\bm{x}_{t-1} \gets \bar{\bm{x}}_{t} + \zeta_t \bm{\kappa}_t$ \hfill $\triangleright$  \textit{Data consistency update} \eqref{eq:our_approx}
         }
      \ENDFOR
      \STATE {\bfseries return} $\hat{\bm{x}}_0$
    \end{algorithmic}
\end{algorithm}

To address this challenge, we propose to approximate the score of the conditional likelihood, $\bm{\kappa}_t = \nabla_{\bm{x}_t} \log p(\bm{y}|\bm{x}_t)$, as 
\begin{equation}
    \label{eq:our_approx}
    % \resizebox{\columnwidth}{!}{
     \bm{\kappa}_t  \simeq \gamma \bm{A}^\mathsf{H}
        (\sigma_n^2 \bm{I}  + \bm{A}\bm{\Sigma}_{0|t}\bm{A}^\mathsf{H} )^{-1} (\bm{y} - \bm{A}\hat{\bm{x}}_0),
    % }
\end{equation}
where $\gamma =   \left(\frac{1}{{\sqrt{\bar{\alpha}}}_t} -  \dfrac{1 - {\bar{\alpha}}_t}{{\sqrt{\bar{\alpha}}_t}(1-\bar\alpha_t + \sigma_0^{2} \bar\alpha_t)}\right) \approx \frac{\partial \hat{\bm{x}}_0}{\partial \bm{x}_t}$, $\bm{\Sigma}_{0|t} = \sigma_{0|t}^2$ as given in \eqref{eq:var_x0_xt_approx}, and $\sigma_0^2 = \mathbb{E}[\bm{x}_0^2]$ is the prior variance that we treat as a hyper-parameter. 
The score of the conditional likelihood comes from the assumption that $p(\bm{x}_0) \sim \mathcal{N}(0, \sigma_0^2 \bm{I})$; therefore, the likelihood $p(\bm{y}|\bm{x}_t)$ can be written as
\begin{equation}
\label{eq:our_pyx_est}
p(\bm{y}|\bm{x}_t) \approx \mathcal{N}(\bm{A}\hat{\bm{x}}_{0}, \sigma_n^2 \bm{I} + \bm{A}\bm{\Sigma}_{0|t}\bm{A}^\mathsf{H}).
\end{equation}
%
% We treat the prior variance $ \sigma_0^2$ .
%
Our proposed score function in \eqref{eq:our_approx} is directly obtained from \eqref{eq:our_pyx_est} (see the derivation in Appendix \ref{sec:cond_score_est}). While $\bm{\kappa}_t$ in \eqref{eq:our_approx} shares certain similarities with \eqref{eq:pigdm_approx} proposed by $\Pi$GDM, \textbf{a key distinction is that our method avoids the computationally intensive Jacobian term ${\partial \hat{\bm{x}}_0}/{\partial \bm{x}_t}$.} We avoid this expensive computation as \textbf{a direct consequence of our assumption on $p(\bm{x}_0)$}, as shown in the following derivation:
\begin{align}
    \notag
     \dfrac{\partial \hat{\bm{x}}_0}{\partial \bm{x}_t} &=\dfrac{\partial}{\partial \bm{x}_t} \left( \frac{1}{{\sqrt{\bar{\alpha}}}_t} \left( \bm{x}_t - \sqrt{1-\bar\alpha_t}\bm{\epsilon}_\theta(\bm{x}_t) \right) \right)
     \\
     \notag
     &\stackrel{(i)}{=} \frac{1}{{\sqrt{\bar{\alpha}}}_t} + \frac{(1 - {\bar{\alpha}}_t)}{{\sqrt{\bar{\alpha}}_t}}  \dfrac{\partial^2}{\partial \bm{x}_t^2} \log p(\bm{x}_t)
     \\
     \label{eq:dx_0_hat_d_xt}
       &\stackrel{(ii)}{=}  \frac{1}{{\sqrt{\bar{\alpha}}}_t} -  \dfrac{1 - {\bar{\alpha}}_t}{{\sqrt{\bar{\alpha}}_t}(1-\bar\alpha_t + \sigma_0^{2} \bar\alpha_t)}.
\end{align}
The first equality $(i)$ comes from the fact that 
\[
\bm{s}_\theta(\bm{x}_t) = \nabla_{x_t} \log p(\bm{x}_t) = {- \bm{\epsilon}_{\theta}(\bm{x}_t)}/{\sqrt{1-\bar{\alpha}_t}}.
\]
The second equality $(ii)$ similarly comes from the Gaussian prior assumption and the second order derivative of log of its density function will be the negative inverse of the variance of $\bm{x}_t$. If we take limit $\sigma_0^2 \to \infty$ in \eqref{eq:dx_0_hat_d_xt}, we readily obtain 
\[
\dfrac{\partial \hat{\bm{x}}_0}{\partial \bm{x}_t} = \dfrac{1}{\sqrt{\bar{\alpha}}_t},
\]
which is equivalent to the expression provided in DDS \cite{chung2024decomposed}. {Similarly, for the noiseless case, our conditional score correction becomes $$
\bm{\kappa} \propto \bm{A}^\mathsf{H} ( \bm{A}\bm{A}^\mathsf{H} )^{-1} (\bm{y} - \bm{A}\hat{\bm{x}}_0) = \bm{A}^\dagger(\bm{y} - \bm{A}\hat{\bm{x}}_0).$$
This expression effectively projects the residual into the row space of $\bm{A}$, similar to DDNM~\cite{wangzero}, which explicitly reconstructs the row-space component using $\bm{A}^\dagger \bm{y}$ and refines the null-space using the denoised estimate $\hat{\bm{x}}_0$.}

As we discussed in section \ref{subsec:prior_approx}, many of the existing methods make an implicit assumption that $p(\bm{x}_0)$ is Gaussian. In this work, we make this assumption explicit and use it to find a tractable and simple model for $p(\bm{y}|\bm{x}_t)$. 

\begin{table}[tbp]
\centering
\caption{Quantitative results for noisy inverse problems on the \textbf{FFHQ} dataset. For each problem, our proposed method delivers the best results in either reconstruction or perceptual quality metrics. (\textbf{Bold} and \underline{underline} indicate the best and second-best results, respectively).}
\resizebox{\textwidth}{!}{
\begin{tabular}{llllllllllllll}
\toprule

{} && \multicolumn{4}{c}{\textbf{Deblur (Gaussian)}} & \multicolumn{4}{c}{\textbf{Deblur (motion)}} & \multicolumn{4}{c}{\textbf{SR ($\times 4$)}} \\

\cmidrule(lr){3-6}
\cmidrule(lr){7-10}
\cmidrule(lr){11-14}

{\textbf{Method}} & {\textbf{NFE}} & 
{PSNR $\uparrow$} & {SSIM $\uparrow$} & {FID $\downarrow$} & {LPIPS $\downarrow$} & 
{PSNR $\uparrow$} & {SSIM $\uparrow$} & {FID $\downarrow$} & {LPIPS $\downarrow$} & 
{PSNR $\uparrow$} & {SSIM $\uparrow$} & {FID $\downarrow$} & {LPIPS $\downarrow$} \\

\midrule

{\thead[l]{DDS \cite{chung2024decomposed}}} 
& {100} 
% Deblur (Gaussian)
&  {27.14}  & {0.777}   & {49.73}  & {0.300}    
% Deblur (motion)
&   {20.48}  &  {0.542} & {87.53} &  {0.443}
% SR (x4)
& {25.63}
&  {0.709} & {66.36}   &  {0.373}   \\

{\thead[l]{ReSample \cite{songsolving}}} 
& {100} 
% Deblur (Gaussian)
&  {25.14}  &  {0.615}   &  {57.92} &   {0.378}    
% Deblur (motion)
& {23.75}   & {0.513} &{0.697}   & {0.474}
% SR (x4)
& {19.33} & {0.292}  & {143.4}   & {0.627}    \\

% {\thead[l]{DDNM \cite{wangzero}}} 
% & {100} 
% % Deblur (Gaussian)
% &    &    &   &      
% % Deblur (motion)
% &    &  &  & 
% % SR (x4)
% &  &   &    &     \\

\thead[l]{$\Pi$GDM \cite{song2022pseudoinverse}} 
& 100 
% Deblur (Gaussian)
&  20.02 & 0.447  & 104.81 & 0.560     
% Deblur (motion)
& \underline{27.41}  & \underline{0.782} &  35.00 & 0.251  
% SR (x4)
& 23.79  & 0.598 & 69.82  & 0.392    \\

\thead[l]{DMPS \cite{meng2022diffusion}} 
& 1000 
% Deblur (Gaussian)
&  25.55 & 0.724  & 28.76 &  0.259    
% Deblur (motion)
& - & -  &  -  & -  
% SR (x4)
& 26.33  & \underline{0.754} & \textbf{26.55} & 0.247       \\

\thead[l]{DPS \cite{chung2023diffusion}} 
& 1000
% Deblur (Gaussian)
& {26.12}& {0.748} & \textbf{26.26} &\underline{0.237} 
% Deblur (motion)
& 23.96 & 0.678 & 29.82 & 0.286
% SR (x4)
& 25.08  & 0.710 &  29.52 & 0.270   \\

\thead[l]{DiffPIR 
\cite{zhu2023denoising}} 
& 100
% Deblur (Gaussian)
& 24.58 & 0.674 & 29.87 & 0.298
% Deblur (motion)
& {26.92}  & {0.757}  & \textbf{25.38} & {0.254}  
% SR (x4)
& 22.96  & 0.666 &  46.83 &0.357  \\

\thead[l]{DDRM~\cite{kawar2022denoising}} 
& 20 
% Deblur (Gaussian)
& 25.90  & 0.741 &  57.88 & 0.303
% Deblur (motion)
& - & - & - & -
% SR (x4)
& \textbf{26.47}  & \textbf{0.761} & 56.41 & 0.299   \\

\thead[l]{MCG~\cite{chung2022improving}}
&  1000 
% Deblur (Gaussian)
& 6.72  & 0.051 & 101.2 & 0.340
% Deblur (motion)
& 6.72  & 0.055 & 310.5 &  0.702
% SR (x4)
& 20.05  & 0.559 & 87.64  & 0.520   \\

\thead[l]{PnP-ADMM~\cite{chan2016plug}}
& - 
% Deblur (Gaussian)
&  23.58 & 0.684 & 94.25 & 0.418
% Deblur (motion)
& 23.43  & 0.669 & 87.23  & 0.450 
% SR (x4)
& 22.14  & 0.592 & 135.24 & 0.530   \\

\thead[l]{ADMM-TV~\cite{goldstein2009split}} 
& - 
% Deblur (Gaussian)
&  24.63 & 0.721 & 101.83 & 0.460
% Deblur (motion)
& 21.95  & 0.655 & 158.43 & 0.519
% SR (x4)
& 21.14  & 0.618 &  258.06 & 0.580   \\

\cmidrule(l){1-14}

\textbf{\thead[l]{CoDPS(Ours)}} 
& 1000 
% Deblur (Gaussian)
& \textbf{27.75} & \textbf{0.802} & \underline{26.54} &\textbf{0.222} 
% Deblur (motion)
& \textbf{27.47} & \textbf{0.787} &  {29.26} &\underline{0.245}
% SR (x4)
& \underline{26.34} & 0.751 &   \underline{27.16} &\underline{0.249}  \\

\textbf{\thead[l]{CoDPS(Ours)}} 
& 100 
% Deblur (Gaussian)
& \underline{27.56} & \underline{0.799} & 33.13 &{0.239} 
% Deblur (motion)
& {27.15} & {0.771} &  \underline{26.00} &\textbf{0.243}
% SR (x4)
& {26.01 } & 0.742 &   
39.74 & {0.289 }  \\

\bottomrule
\end{tabular}
}
\label{tab:re_ffhq_sr_deblur}

\end{table}

\subsection{Covariance matrix inversion}
\label{subsec:eff_cov_inv}
Our algorithm requires inverting the covariance matrix of the conditional distribution $p(\bm{y}|\bm{x}_t)$, given as $(\sigma_n^2 \bm{I} + \sigma^2_{0|t}\bm{A}\bm{A}^\mathsf{H})$. While inverting such a large covariance matrix directly is not always feasible, below we discuss some efficient methods to invert this matrix for a number of inverse problems. Note that for a forward operator with orthogonal rows, we can write $\bm{A}\bm{A}^\mathsf{H}$ as an identity matrix and simplify the inversion. The details and derivation for each problem can be found in Appendix \ref{sec:adx_cov_inv}.

\noindent\textbf{Inpainting.}
For inpainting, we can write the forward model as 
$\bm{y} = \mathcal{A}(\bm{x}) + \bm{\eta} = \bm{M} \odot \bm{x} + \eta,$
where $\bm{M}$ denotes an inpainting mask with ones and zeros for observed and unobserved pixels, respectively. The conditional score can be written as 

\begin{equation*}
 \bm{\kappa}_t = \gamma \bm{M} \odot \left(\frac{ \bm{y} - \bm{M}\odot \hat{\bm{x}}_0}{\sigma_n^2 \bm{1} + \sigma^2_{0|t}\bm{M} }\right).
\end{equation*}

We also note that the division is element-wise and $\bm{1}$ is a matrix of all ones with the same size as $\bm{M}$. 

\noindent\textbf{Deblurring.} For deblurring we can write the forward model as 
$
 \bm{y} = \bm{h} \circledast \bm{x} + \bm{\eta} = \bm{H}\bm{x} + \bm{n}.
$
We assume the blurring matrix $\bm{H}$ is the matrix representation of the cyclic convolution $ \bm{h} \circledast \bm{x}$. Under this assumption, we can represent $H$ as a doubly blocked block circulant matrix (BCCB). 

The conditional score can be written as 
\begin{equation*}
 \bm{\kappa}_t = \gamma \bm{F}^\mathsf{H}\bm{\Lambda}^\mathsf{H}\left(\frac{1}{\sigma_n^2}\bm{I} + \frac{1}{\sigma^2_{0|t}|\bm{\Lambda}|^2} \bm{I}\right)(\bm{F}\bm{y} - \bm{\Lambda} \bm{F}\bm{\hat{x}_0}),
\end{equation*}
where $\bm{F},\bm{F}^\mathsf{H}$ denote the forward and inverse/adjoint 2D Fourier transform operators, respectively; $\bm{\Lambda} = \bm{F}\bm{h}$ represents a vector with the 2D Fourier transform of the blur kernel $\bm{h}$.

\noindent\textbf{Super-resolution.} For super-resolution, we can write the forward model as 
$
 \bm{y} = (\bm{h} \circledast \bm{x})_{\downarrow_d} + \bm{\eta} = \bm{S}\bm{H}\bm{x} + \bm{n},
$
where $\bm{S}$ is the down-sampling matrix with a factor $d$, $\bm{H}$ is the cyclic-convolution matrix obtained from the blur kernel $\bm{h}$.  The down-sampling operator is a decimation matrix that samples the first pixel in every non-overlapping $d \times d$ block of our image. 

The conditional score can be written as 
\begin{equation*}
 \bm{\kappa}_t = \gamma \bm{F}^\mathsf{H}\bm{\Lambda}^\mathsf{H}\bm{F}\bm{S}^\mathsf{H}\left({\frac{1}{\sigma_n^2} \bm{I} + \frac{1}{\sigma^2_{0|t}|\bm{\Gamma}|^2}\bm{I}}\right)(\bm{y} - \bm{S}\bm{H}\bm{\hat{x}_0}),
\end{equation*}
where $\bm{\Gamma}$ is a vector of $m^2$ elements obtained by periodic averaging of $d^2$ entries in the 2D Fourier transform of the blur kernel. 

\noindent\textbf{Separable systems.} Given a separable system with measurements $\bm{y} = \mathcal{A}(\bm{x}) + \bm{\eta}$, we can represent the measurements and forward operator as 
$\bm{Y}$ and $\mathcal{A}(x) = \bm{A}_l\bm{X}\bm{A}_r^\mathsf{H}$. 
$\bm{X} \in \mathbb{R}^{n\times n}$ and  
$\bm{Y} \in \mathbb{R}^{m\times m}$ represent reshaped 2D versions of $\bm{x} \in \mathbb{R}^{n^2}$ and $\bm{y} \in \mathbb{R}^{m^2}$, respectively, 
and $\bm{A}_l, \bm{A}_r \in \mathbb{R}^{m \times n}$ represent the left and right matrices. We can directly write the conditional score \eqref{eq:our_approx} as
\begin{equation*}
 \bm{\kappa}_t = \gamma   \bm{V}_l \bm{\Sigma}_l \left(\frac{\bm{U}_l^\mathsf{H} (\bm{Y} - \bm{A}\hat{\bm{X}}_0) \bm{U}_r}{\sigma_n^2 \bm{1}\bm{1}^\mathsf{H} + \bm{\sigma}^2_{0|t} (\bm{\sigma_l}\bm{\sigma}_r^\mathsf{H})^2 } \right) \bm{\Sigma}_r \bm{U}_r^\mathsf{H}.
\end{equation*}

Here, $\bm{Y}$ and $\hat{\bm{X}}_0$ are the matrix forms of $\bm{y}$ and $\hat{\bm{x}}_0$, respectively. The matrices $\bm{U}_l$, $\bm{\Sigma}_l$, $\bm{V}_l$ and $\bm{U}_r$, $\bm{\Sigma}_r$, $\bm{V}_r$ represent the SVD of $\bm{A}_l$ and $\bm{A}_r$, with diagonal entries $\bm{\sigma}_l$ and $\bm{\sigma}_r$.

% -------------------------------------------------------------
\begin{figure}[t]
        \centering
    \includegraphics[width=0.65\columnwidth]{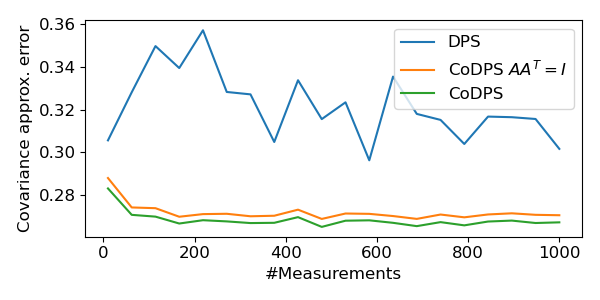}
        \caption{CoDPS provides a more accurate approximation of the correct MAP estimators. A simplified version of CoDPS, assuming $\bm{A}\bm{A}^\mathsf{H} = \bm{I}$, performs comparably to the CoDPS. This shows the effectiveness of our method in cases where the covariance matrix cannot be inverted. }
        \label{fig:cov_approx_err}
\end{figure}
\begin{figure*}[ht]
    \centering
    \includegraphics[width=\textwidth]{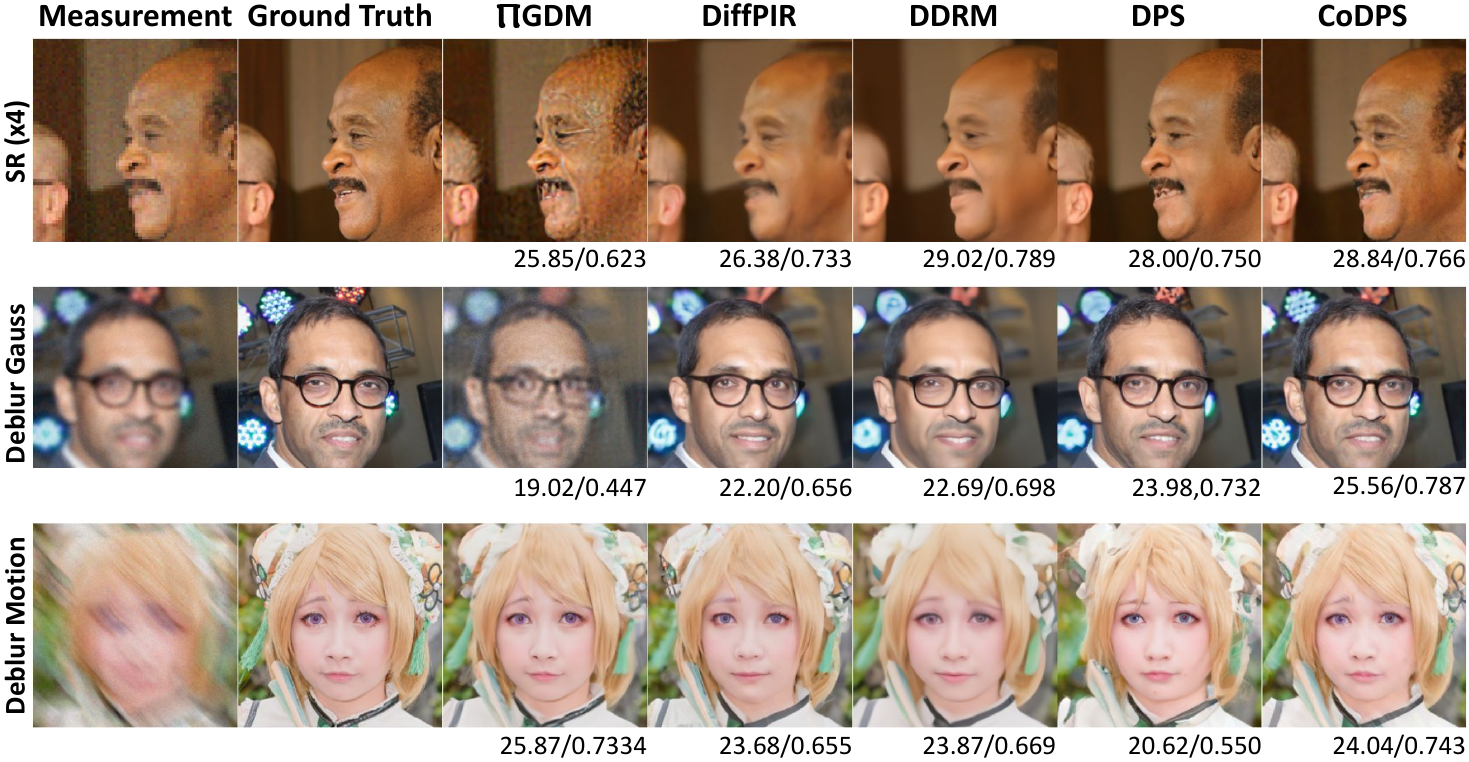}
    \caption{Our method (CoDPS) outperforms most competing approaches both quantitatively and qualitatively. For super-resolution, our method's output recovers facial textures that are missing in the other outputs. For Gaussian and motion deblurring, our method consistently recovers background and foreground details, and produces a subject pose that is more consistent with the ground truth. (PSNR/SSIM values are shown below each image.)
    }
    \label{fig:comp_sr_deblur}
\end{figure*}

% \begin{figure*}[t]
%     \centering
%     \includegraphics[width=0.7\textwidth]{figures/comp_motion.pdf}
%     % 
%     \caption{Motion deblurring results on FFHQ show that, except for $\Pi$GDM, our method outperforms competing methods. \protect\footnotemark
%     % \footnote{12}
%     \vspace{-1em}
%     }
%     \label{fig:comp_motion_deblur}
% \end{figure*}
\section{Experiments}

\subsection{Experiment on mixture of gaussians}
\label{sebsec:exp_on_mix_gaus}
We first present an experiment for solving an ill-posed inverse problem \eqref{eq:fwd_model} using synthetic data. We define the prior distribution $p(\bm{x}_0)$ using a Gaussian mixture model (GMM), and the forward operator as a random projection matrix $\bm{A} \in \mathbb{R}^{m \times n}$, where $m < n$. First, we use a diffusion network to learn the data prior. Then, we apply our proposed method to solve the inverse reconstruction problem.  The primary objective of this experiment is to highlight the importance of our covariance correction term by comparing it with DPS. We present the details of this experiment in Appendix \ref{sebsec:adx_exp_on_mix_gaus}.

 Our experimental validation indicates that our proposed method yields better reconstruction performance and better approximates the true posterior covariance. Figure \ref{fig:cov_approx_err} shows the approximation error between the true MAP estimator's covariance and the sample covariance estimates. The idea is that the sample covariance of the reconstructed estimations for each cluster in the GMM should approximate and converge to the covariance of the MAP estimator. We observe that both our methods yield better approximation that improves as we obtain more measurements. The simplified CoDPS ($\bm{A}\bm{A}^\mathsf{H}=\bm{I}$) performs better than DPS and is  comparable to the version of our method that uses the correct covariance. Since the simplified version of our method accounts for the uncertainty in estimates $\hat{\bm{x}}_0$ \eqref{eq:var_x0_xt_approx}, we observed an improved result over DPS. 

\subsection{Experiment on natural images. }

We performed image restoration experiments using the FFHQ \cite{karras2019style} and ImageNet \cite{deng2009imagenet} datasets. Our experiments include image super-resolution, random inpainting, and Gaussian and motion deblurring tasks. 
We compare our method with related diffusion-based approaches such as $\Pi$GDM \cite{song2022pseudoinverse}, DDRM \cite{kawar2022denoising}, DiffPIR \cite{zhu2023denoising}, DPS \cite{chung2023diffusion}, DMPS \cite{meng2022diffusion} , MCG \cite{chung2022improving}, {DDS\cite{chung2024decomposed}, DDNM\cite{wangzero} and ReSample\cite{songsolving}}. Additionally, we compare it with ADMM for Total Variation (TV) regularization \cite{rudin1992nonlinear, goldstein2009split} and PnP-ADMM \cite{chan2016plug} using a DnCNN \cite{zhang2017beyond} denoiser. For each of these methods, we used the publicly available source code with the reported hyper-parameters. Further details of the experimental setup, hyper-parameters for our method, as well as for all the comparison methods, are described in Appendix \ref{sec:imp_details}.

\noindent
\textbf{Experimental setup.} We utilized existing pre-trained diffusion models for both datasets, for FFHQ we obtained it from \cite{chung2023diffusion} and for ImageNet from \cite{dhariwal2021diffusion}. Note that both of these network are trained as generative denoisers for their respective datasets and are not fine-tuned for any image restoration task. We use these models for all methods in our comparison set that use diffusion models. Following \cite{chung2023diffusion}, we use the first 1K images of the FFHQ dataset as validation data. We resize the original $1024 \times 1024$ images in this dataset to $256 \times 256$. For ImageNet, we obtain the preprocessed $256 \times 256$ images from \cite{dhariwal2021diffusion} and use the first 1K images as our validation dataset. We use these validation datasets for all comparison experiments. For performance metrics, we report both standard reconstruction metrics, including peak signal-to-noise ratio (PSNR) and structural similarity index (SSIM) \cite{wang2004image}, as well as perceptual metrics, including Fréchet inception distance (FID) \cite{heusel2017gans} and Learned Perceptual Image Patch Similarity (LPIPS) \cite{zhang2018unreasonable}. We used PyTorch \cite{paszke2019pytorch} to implement our proposed method on a single NVIDIA GeForce RTX 2080 Ti GPU with 12GB memory.

\noindent
\textbf{Inverse problems. }
The first inverse problem we consider is image super-resolution. The task here is to recover an image blurred using a $9 \times 9$ Gaussian blur kernel with standard deviation $3.0$ followed by a $\times 4$ down-sampling using a decimation matrix. 
The down-sampling operator is a decimation matrix that samples the first pixel in every non-overlapping $4 \times 4$ block of our image. 
For Gaussian deblurring, we use a $61 \times 61$ kernel with standard deviation $3.0$ and for motion deblurring we randomly generate a kernel\footnote{{https://github.com/LeviBorodenko/motionblur}} with size $61 \times 61$ with intensity of $0.5$. Finally, we consider random image inpainting, where we remove pixels randomly in all color channels. Each pixel can be removed with a uniform probability in the range $[0.7, 0.8]$. For noisy experiments, we apply an additive Gaussian noise with standard deviation $\bm{\sigma}_n=0.05$.

\begin{table*}[t]
\centering
\resizebox{\textwidth}{!}{
\begin{tabular}{llllllllllllll}
\toprule

{} & & \multicolumn{4}{c}{\textbf{SR ($\times 4$)}} 
& 
\multicolumn{4}{c}{\textbf{Deblur (Gaussian)}} & \multicolumn{4}{c}{\textbf{Deblur (motion)}}\\

\cmidrule(lr){3-6}
\cmidrule(lr){7-10}
\cmidrule(lr){11-14}

% \cmidrule(lr){10-11}
% {\textbf{Method}}  & {\textbf{NFE}} & {\textbf{Wall-clock time} / img} & -
{\textbf{Method}}  & {\textbf{NFE}}  &

{PSNR $\uparrow$} & {SSIM $\uparrow$} & {FID $\downarrow$} & {LPIPS $\downarrow$} & {PSNR $\uparrow$} & {SSIM $\uparrow$} & {FID $\downarrow$} & {LPIPS $\downarrow$}  & 
{PSNR $\uparrow$} & {SSIM $\uparrow$} & {FID $\downarrow$} & {LPIPS $\downarrow$} \\

\midrule

{\thead[l]{RED-DIFF \cite{mardanivariational}}} 
& {100} 
& {25.97}  & {0.755}  & {39.98}    & { 0.275 }  
% Deblur (Gaussian)
&  {26.17}  & {0.746}   &  {45.57} & {0.322}      
% Deblur (motion)
&  {29.60}  & {0.862} & {30.30} & {0.218}    \\

{\thead[l]{DDNM \cite{wangzero}}} 
& {100} 
& {26.00}  & {0.741}  & {38.21}    & {0.273}  
% Deblur (Gaussian)
&  {24.85}  & {0.661}   &  {46.75} & {0.338}      
% Deblur (motion)
&  -  & - &  - & -    \\
% SR (x4)

\thead[l]{$\Pi$GDM} \cite{song2022pseudoinverse}
% wall clock times (37.50,36.337,37.767)
% SR
& 100 &   \textbf{26.13} & \textbf{0.757} & \textbf{26.92} & \textbf{0.239}

& 24.74 & 0.683 & 30.31 & 0.271   
% Motion Deblur
& \textbf{45.21} & 0.983 & 5.473 &  0.05  \\

\textbf{\thead[l]{CoDPS(Ours)}} 
% SR (19.956,18.6959,18.7817)
& 100 &   24.84 & 0.726 & 42.42 & 0.284  
% % Inp random
% & 23.77  & 0.727 & & 0.360 
% Gaus Deblur
&  \textbf{29.10} &  \textbf{0.875} & \textbf{12.31}  &   \textbf{0.130}
% Motion Deblur
& 43.39 & \textbf{0.992} & \textbf{0.622} &   \textbf{0.01} \\

\bottomrule
\end{tabular}
}
\caption{Quantitative metrics for noiseless inverse problems on the \textbf{ImageNet} dataset. Our proposed method is approximately twice as computationally efficient as $\Pi$GDM and achieves the best performance in deblurring tasks.}
\label{tab:imgnet_pgdm_codps}
\end{table*}

\noindent\textbf{Results. }We present quantitative metrics on the FFHQ dataset in Tables \ref{tab:re_ffhq_sr_deblur} and \ref{tab:re_ffhq_inp}. In both motion and Gaussian deblurring experiments, our methods with 1000 and 100 number of function evaluations (NFEs) outperformed all competing methods across several metrics. With the exception of the FID scores, our methods consistently delivered either the best or second-best performance. The only two exceptions are: DPS achieved the top FID score for Gaussian deblurring, while DiffPIR obtained the best FID score for motion deblurring. These results demonstrate the effectiveness of our approaches in both deblurring tasks.
For $4\times$ super-resolution experiments, our method demonstrates strong performance, achieving a PSNR of 26.34 dB and an SSIM of 0.751. While DDRM achieves a marginally higher PSNR (+0.13 dB), CoDPS demonstrates superior perceptual quality. When compared to DMPS, although CoDPS shows lower perceptual quality, its performance remains highly competitive with comparable PSNR, SSIM and LPIPS metrics.

\begin{wraptable}{r}{0.5\textwidth}
\vspace{-1.5em}
% \begin{table}[htbp]
\centering
\caption{Random image inpainting results on \textbf{FFHQ} validation dataset. Our methods outperformed several competing approaches, with the exception of DPS. }
\resizebox{0.5\textwidth}{!}{
\begin{tabular}{lllll}
\toprule

{} & \multicolumn{4}{c}{\textbf{Inpainting (random)}}
\\
\cmidrule(lr){2-5}
% \cmidrule(lr){10-11}
{\textbf{Method}} & {PSNR $\uparrow$} & {SSIM $\uparrow$} & {FID $\downarrow$} & {LPIPS $\downarrow$} \\
\midrule
\thead[l]{DPS \cite{chung2023diffusion}} % SR
% Inp random
&  \textbf{29.12} &  \textbf{0.852} & \textbf{27.81}& \textbf{0.195} \\

{\thead[l]{ReSample 
\cite{songsolving}}}  & {27.78} & {0.795} & {54.52} &  {0.274}\\

{\thead[l]{DDS \cite{chung2024decomposed}}}
& {27.01}  &  {0.743} &  {46.99}  &  {0.325} \\

{\thead[l]{DDNM \cite{wangzero}}} 
&  {\underline{28.69}}   &  {\underline{0.840}}  & {36.84}  &  {\underline{0.218}}   \\

\thead[l]{DMPS \cite{meng2022diffusion}} 
% Inp random
&  21.33  &  0.573  &  112.30 & 0.573  \\
\thead[l]{DDRM~\cite{kawar2022denoising}} % SR
% Inp random
& 25.26  & 0.761 & 61.57 & 0.295  \\
\thead[l]{MCG~\cite{chung2022improving}}
% Inp random
& 26.59  & 0.778 & \underline{31.94} & 0.231  \\
\thead[l]{PnP-ADMM~\cite{chan2016plug}}
% Inp random
& 23.77  & 0.727 & 57.00 & 0.360 \\
\thead[l]{ADMM-TV~\cite{goldstein2009split}}
% Inp random
&  20.50 & 0.623 & 141.45 & 0.568 \\
\cmidrule(l){1-5}

\textbf{\thead[l]{CoDPS(NFE 1000)}} 
% Inp random
&  {27.56} & {0.802}  &  33.19 & {0.222}
\\
\textbf{\thead[l]{CoDPS(NFE 100)}} 
% Inp random
&  {27.32} &  {0.795}  &  37.87 & {0.236}
\\

\bottomrule
\end{tabular}
}
\label{tab:re_ffhq_inp}
% \end{table}
\vspace{-1em}
\end{wraptable} 

We present example outputs for these experiments in Figures~\ref{fig:comp_sr_deblur} . The outputs of our method demonstrate both high visual quality and consistency with the measurements. This behavior is even more pronounced in random image inpainting experiments. From Table \ref{tab:re_ffhq_inp}, we observe that DPS is the best-performing method quantitatively {and DDNM ranks second}. {Our method remains competitive to other methods such as ReSample and DDS}. Our method outputs reconstructions that are more consistent with the ground truth and have better visual quality. Figure \ref{fig:ffhq_inp} shows example outputs from the inpainting experiments. Our method is able to recover fine details, which is easily observed in the zoomed versions. 
In the first image, the background pattern is missing in both DDRM and DPS outputs, while our method successfully recovers it. Similarly, in the second image, the outputs of the other methods are overly smooth and lack details.  

Similar to the results on the FFHQ dataset, our method consistently achieves either the top or second-best performance on the ImageNet dataset across various inverse problems, as shown in Table \ref{tab:results_linear_imgnet_psnrssimlpips}. For Gaussian deblurring, our method surpasses all others in PSNR, FID, and LPIPS. In motion deblurring, CoDPS demonstrates superior perceptual quality. In image super-resolution, our method achieves the second-best performance across all metrics. For random image inpainting, our method is competitive with DPS and outperforms other methods, as shown in Table \ref{tab:results_linear_imgnet_inpaint}. We provide sample outputs in Appendix \ref{subsec:apdx_add_figs}.

{We designed an additional noiseless experiment to compare our method with $\Pi$GDM and DDNM}, and we report the results in Table \ref{tab:imgnet_pgdm_codps}. From the noisy experiments, we observed that both {methods $\Pi$GDM, RED-DIFF and DDNM+} did not perform well. As a result, we designed an experiment that closely matches the setup in the original papers. We used the ImageNet dataset on three tasks: super-resolution, motion deblurring, and Gaussian deblurring. We have also included RED-Diff in this evaluation.
Our method outperformed all competing methods in the Gaussian deblurring experiment, while competing methods performed better in the super-resolution experiment. For motion deblurring, our method achieved the best metrics in terms of SSIM, FID, and LPIPS, while $\Pi$GDM achieved the best PSNR scores. {We did not include motion deblurring results because the current implementation of DDNM does not support it.} Overall, our method produced consistent and improved results in both noiseless and noisy problems.

\begin{wraptable}{r}{0.5\textwidth}
\centering
\caption{Comparison of efficiency metrics for various diffusion-based solvers on SR $(\times 4)$ using the FFHQ dataset, highlighting the number of function evaluations (NFE), wall-clock time in seconds and peak GPU memory usage when performing inference on a single $256\times 256$ RGB image. }
\resizebox{0.5\columnwidth}{!}{
\begin{tabular}{llcc}
\toprule \\

 Method & NFE &  \makecell{\textbf{Wall-clock time}  \\ \textbf{ (Secs/img)}}  $\downarrow$ &  \makecell{\textbf{GPU Memory Req} \\ \textbf{(in GB)} $\downarrow$ }\\

\midrule

{\thead[l]{RED-DIFF}} 
& {1000} 
& {101.8}  & {9.44} \\

{\thead[l]{DDS} }& {1000} & {45.6} & {4.73} \\
{\thead[l]{DiffPIR} }& {1000} & {45.7} & {\textbf{4.22}} \\
{\thead[l]{MCG} }& {1000} & {84.0} & {9.56} \\
{\thead[l]{DDRM} }& {1000} & {54.7} & {4.73} \\

{\thead[l]{ReSample}}
% wall clock times (37.50,36.337,37.767)
% SR
& {100} & {230.4}  &  {\textbf{4.32}} \\
{\thead[l]{DDNM}}
% wall clock times (37.50,36.337,37.767)
% SR
& {1000} & {45.73}  &  {4.73} \\ 
\thead[l]{DPS}
% wall clock times (37.50,36.337,37.767)
% SR
& 1000 &  146.8  & 9.35 \\ 
\thead[l]{DMPS}
% wall clock times (37.50,36.337,37.767)
% SR
& 1000 &  \underline{63.9} & \underline{4.33} \\ 
\thead[l]{$\Pi$GDM}
% wall clock times (37.50,36.337,37.767)
% SR
& 1000 &  81.4 & 9.90  \\

\textbf{\thead[l]{CoDPS(Ours)}} 
% SR (19.956,18.6959,18.7817)
& 1000 & \textbf{42.7} & {4.73}   \\

\bottomrule
\end{tabular}
}
\label{tab:performance_comp}
% \vspace{-1.5em}
\end{wraptable}

\noindent\textbf{Algorithm efficiency.} Our method is both memory and time efficient compared to related methods such as DPS and $\Pi$GDM. Both these methods need to compute gradients through the diffusion networks, which can increase the cost per NFE by a factor of $2$-$3$ as discussed in \cite{song2022pseudoinverse}. 
In Table \ref{tab:performance_comp}, we report the average wall-clock time per image in seconds and the peak GPU memory requirement during inference for each method. For a consistent NFE of 1000, our method, CoDPS, is approximately $2\times$ faster on average compared to DPS, achieving a wall-clock time of 42.7 seconds per image. 
Additionally, CoDPS shows competitive memory efficiency to DDS, DiffPIR, DMPS and {ReSample}, requiring only 4.73 GB of GPU memory. {However, the runtime of ReSample is substantially worse compared to our method. ReSample uses a latent diffusion-based model, where the latent image size is 64. This provides some advantage in terms of computation, as the diffusion process occurs over a smaller size ($64\times 64$) compared to the pixel space ($256\times 256$). However, the decoding and ReSampling stages can be computationally expensive, which resulted in a significantly longer runtime compared to our method.} Similar to DMPS and DDS, our method avoids gradient computations through the diffusion networks, leading to improved memory efficiency. Moreover, CoDPS significantly outperforms both DPS and $\Pi$GDM in memory efficiency. These findings show the effectiveness of CoDPS in optimizing computational resources and make it suitable for resource constrained settings.

\section{Conclusion}
In this paper, we proposed a unified diffusion posterior sampling method for solving inverse problems that is efficient and more accurate. Our method uses existing implicit and explicit assumptions made in prior works. We leverage measurement covariance correction to improve reconstruction performance, and the prior Gaussian assumption to avoid computing backpropagation through the pre-trained diffusion network.
We demonstrate the performance of our methods in several inverse problems. Additionally, we present efficient techniques for covariance matrix inversion in applications such as image inpainting, image deblurring, super resolution, and separable systems.

\textbf{Acknowledgement.} This paper is partially based on work
supported by the NSF CAREER awards under grants CCF-2043134 and CCF-2046293. 

% \clearpage
\bibliographystyle{unsrt}
\bibliography{ref}

\appendix
\section*{\large \bf Appendix: Supplementary Material}

\section{Proofs}
\subsection{Posterior mean}
\label{subsec:posterior_mean_adx}

In this subsection, we show that equation \eqref{eq:p_x0_xt_approx} satisfies the unique posterior mean \eqref{eq:mean_p_x0_xt} as derived from Tweedie’s identity. The prior distribution, $ p(\bm{x}_0) $, is modeled as a Gaussian distribution $ \mathcal{N}(\bm{0}, \sigma_{0}^2 \bm{I}) $. The likelihood distribution $ p(\bm{y}|\bm{x}_0) $ is given by $ \mathcal{N}(\mathcal{A}\bm{x}_0, \sigma_n^2 \bm{I}) $, and the diffusion forward model $ p(\bm{x}_t|\bm{x}_0) $ is represented as $ \mathcal{N}(\sqrt{\bar\alpha_t} \bm{x}_0, (1-\bar\alpha_t) \bm{I}) $. By applying the properties of marginal and conditional Gaussians \cite{bishop2007}, we derive the following distributions:

\begin{align}
    \label{eq:marg_p_xt}
    p(\bm{x}_t) &= \mathcal{N}(0, (1-\bar\alpha_t + \sigma_{x_0}^{2} \bar\alpha_t) \bm{I})
\end{align}
\begin{align}
     \label{eq:apdx_marg_p_xo_xt}
    p(\bm{x}_0|\bm{x}_t) &= \mathcal{N}\left(\bm{\Sigma}_{0|t}\left(\dfrac{\sqrt{\bar\alpha_t}}{1-\bar\alpha_t} \bm{x}_t \right), \bm{\Sigma}_{0|t}\right),
\end{align}
where
\begin{align}
\label{eq:apdx_marg_p_xo_xt_var}
\bm{\Sigma}_{0|t} = \dfrac{\sigma_0^{2}(1-\bar\alpha_t)}{(1-\bar\alpha_t) + \sigma_0^{2}\bar\alpha_t} \bm{I} = \sigma^2_{0|t}\bm{I}.
\end{align}

Then, the posterior mean is given by
\begin{align*}
    \label{eq:mean_p_x0_xt_adx}
    \hat{\bm{x}}_0 &\approx \dfrac{1}{\sqrt{\bar\alpha_t}}\left( \bm{x}_t - \sqrt{1-\bar\alpha_t}\bm{\epsilon}_\theta(\bm{x}_t) \right) \\
    &=  \dfrac{1}{\sqrt{\bar\alpha_t}}\left( \bm{x}_t + {(1-\bar\alpha_t)} \nabla_{x_t} \log p(\bm{x}_t \right)\\
    &= \dfrac{1}{\sqrt{\bar\alpha_t}}\left( \bm{x}_t - {(1-\bar\alpha_t)} \dfrac{\bm{x}_t}{(1-\bar{\alpha}_t+\sigma_{x(0)}^2\bar{\alpha}_t)} \right)\\
    &= \dfrac{1}{\sqrt{\bar\alpha_t}}\left(\dfrac{\sigma_{x(0)}^2\bar{\alpha}_t\bm{x}_t}{(1-\bar{\alpha}_t+\sigma_{x(0)}^2\bar{\alpha}_t)} \right)\\
    &=  \dfrac{\sigma_{x(0)}^2\sqrt{\bar{\alpha}}_t\bm{x}_t}{(1-\bar{\alpha}_t+\sigma_{x(0)}^2\bar{\alpha}_t)} , 
    % &\approx 
\end{align*}

\noindent which is equivalent to the expectation given in \eqref{eq:apdx_marg_p_xo_xt} after substituting in \eqref{eq:apdx_marg_p_xo_xt_var}.

\subsection{Proof of theorem 1}

\label{subsec:proof_theo_1}
First we will state necessary tools that we will use to prove our theorem. 
\begin{lemma}
\label{lemma:bivariate_g}
    The joint density of a bivariate normal distribution f(x,y) having normal marginals and normal conditionals can be written as 
    \[
f(x, y) = \exp \left( \begin{pmatrix} 1 & x & x^2 \end{pmatrix} 
\underset{=\bm{M}}{
\begin{pmatrix} 
m_{00} & m_{01} & m_{02} \\ 
m_{10} & m_{11} & m_{12} \\ 
m_{20} & m_{21} & m_{22} 
\end{pmatrix} }
\begin{pmatrix} 
1 \\ 
y \\ 
y^2 
\end{pmatrix} \right),
\]
and satisfies the following conditions $m_{21}=m_{12}=m_{22}=0, m_{20}<0$, $m_{02}<0$. If these conditions are met, the statement $m_{11}^2<4m_{02}m_{20}$ will hold if $f(x,y)$ is a valid density function \cite{castillo1989conditional}. 
\begin{proof}
The bivariate normal distribution for random variables $X$ and $Y$ is given by the following probability density function, where $\mu_X$ and $\mu_Y$ are the means, $\sigma_X$ and $\sigma_Y$ are the standard deviations, and $\rho$ is the correlation coefficient between $X$ and $Y$:
\begin{equation}
\label{eq:adx_bv_gauss}
f_{X,Y}(x, y) = \frac{1}{C} \exp \left(-\frac{z}{2(1 - \rho^2)}   \right),
\end{equation}
where
\begin{equation*}
    z =  \left[ \frac{(x - \mu_X)^2}{\sigma_X^2} + \frac{(y - \mu_Y)^2}{\sigma_Y^2} - \frac{2\rho(x - \mu_X)(y - \mu_Y)}{\sigma_X \sigma_Y} \right],
\end{equation*}
and $C = 2 \pi \sigma_X \sigma_Y \sqrt{1 - \rho^2}$. We can rewrite \eqref{eq:adx_bv_gauss} as
\begin{equation*}
f_{X,Y}(x, y) = \exp \left(-\log C -\frac{z}{2(1 - \rho^2)}   \right).
\end{equation*}
We observe that this exponent is quadratic in $x$ and $y$ and we express it as a bi-linear form using a matrix $\bm{M}$. Since the exponent in \eqref{eq:adx_bv_gauss} does not any terms that contain $x^2y^2$,$xy^2$,and $x^2y$, we must have $m_{21}=m_{12}=m_{22}=0$. We can also easily see that $m_{20} = \dfrac{-1}{2(1-\rho^2)\sigma^2_x}$ , $m_{02}= \dfrac{-1}{2(1-\rho^2)\sigma^2_y}$, $m_{11} = \dfrac{-\rho}{(1-\rho^2)\sigma_x\sigma_y} $, and $m_{11}^2 = 4 \rho^2 m_{20}m_{10}$. Since $0 < \rho^2 < 1$, we have $m_{11}^2 < 4 m_{20}m_{10}$. 
\end{proof}

\end{lemma}
\begin{theorem}[\cite{aczel1966lectures}]
\label{theorem:aczel}
Let $f_k(x)$ and $g_k(y)$ be functions  and all solutions of the equations

 \begin{equation}
     \sum_{k=1}^{n} f_k(x)g_k(y) = 0
 \end{equation}
can be written in the form 
\begin{equation}
    \begin{pmatrix}
    \label{eq:aczel_1}
f_1(x) \\
f_2(x) \\
\vdots \\
f_n(x)
\end{pmatrix}
=
\begin{pmatrix}
a_{11} & a_{12} & \cdots & a_{1r} \\
a_{21} & a_{22} & \cdots & a_{2r} \\
\vdots & \vdots & \ddots & \vdots \\
a_{n1} & a_{n2} & \cdots & a_{nr}
\end{pmatrix}
\begin{pmatrix}
\varphi_1(x) \\
\varphi_2(x) \\
\vdots \\
\varphi_r(x)
\end{pmatrix}
\end{equation}
and 
\begin{equation}
    \label{eq:aczel_2}
    \begin{pmatrix}
        g_1(y) \\
        g_2(y) \\
        \vdots \\
        g_m(y)
        \end{pmatrix}
        =
        \begin{pmatrix}
        b_{1(r+1)} & b_{1(r+2)} & \cdots & b_{1n} \\
        b_{2(r+1)} & b_{2(r+2)} & \cdots & b_{2n} \\
        \vdots & \vdots & \ddots & \vdots \\
        b_{m(r+1)} & b_{m(r+2)} & \cdots & b_{mn}
        \end{pmatrix}
        \begin{pmatrix}
        \Psi_{r+1}(y) \\
        \Psi_{r+2}(y) \\
        \vdots \\
        \Psi_n(y)
        \end{pmatrix},
        \end{equation}
            
    where $r$ is an integer between $0$ and $n$,  $\varphi_i(x)$ and   $\Psi_j(y)$ are arbitrary independent functions, while the each $a_{ij}$ and $b_{ij}$ are constants that satisfy 
    \begin{equation}
    \label{eq:aczel_3}
        \begin{pmatrix}
    a_{11} & a_{12} & \cdots & a_{1r} \\
    a_{21} & a_{22} & \cdots & a_{2r} \\
    \vdots & \vdots & \ddots & \vdots \\
    a_{n1} & a_{n2} & \cdots & a_{nr}
    \end{pmatrix}
    \begin{pmatrix}
    b_{1(r+1)} & b_{1(r+2)} & \cdots & b_{1n} \\
    b_{2(r+1)} & b_{2(r+2)} & \cdots & b_{2n} \\
    \vdots & \vdots & \ddots & \vdots \\
    b_{n(r+1)} & b_{n(r+2)} & \cdots & b_{nn}
    \end{pmatrix}
    =
    \bm{0}
    \end{equation}
\end{theorem}

Now that we have stated all the necessary theorems, we can restate Theorem \ref{theorem:cond_gauss} and prove it. A more general version of this theorem is stated in Theorem 3.1 of \cite{arnold1999conditional}. Our theorem here is a specific case and we prove it here for completeness.

% Theorem 3.1 from the Book
\condgauss*
\begin{proof}

    We begin by writing the probability density of the conditional distributions $p(X|Y)$ and $p(Y|X)$ as 
    \begin{equation}
       f_{X|Y}(x|y) = \frac{1}{\sqrt{2\pi \sigma_{X|Y}^2}} \exp \left( -\frac{(x - \mu_{X|Y})^2}{2 \sigma_{X|Y}^2} \right)
    \end{equation}
    and \begin{equation}
       f_{Y|X}(y|x) = \frac{1}{\sqrt{2\pi \sigma_{Y|X}^2}} \exp \left( -\frac{(y - \mu_{Y|X})^2}{2 \sigma_{Y|X}^2} \right).
    \end{equation}

    We have the following identity from Bayes' rule $p(Y|X)p(X) = p(X|Y)p(X)$, which gives us the following equality
    \begin{equation}
\frac{f_X(x)}{\sigma_{Y|X}} \exp \left( -\frac{(y - \mu_{Y|X})^2}{2 \sigma_{Y|X}^2} \right) = \frac{f_Y(y)}{\sigma_{X|Y}} \exp \left( -\frac{(x - \mu_{X|Y})^2}{2 \sigma_{X|Y}^2} \right)
    \end{equation}
    By taking log of both sides of the equations, we obtain 
    \begin{equation}
    \label{eq:u_x_v_y}
    u(x) - \frac{(y - \mu_{Y|X})^2}{2 \sigma_{Y|X}^2} = v(y) -  \frac{(x - \mu_{X|Y})^2}{2 \sigma_{X|Y}^2}  ,
    \end{equation}

    where $ u(x) = \log (f_X(x)) - \log(\sigma_{Y|X})$ and $v(y) = \log (f_Y(y)) - \log(\sigma_{X|Y})$. We can rearrange the terms as 
     \begin{equation}
    \begin{aligned}
        &\sigma_{X|Y}^2[2\sigma_{Y|X}^2 u(x) - \mu_{Y|X}^2] 
        + \sigma_{Y|X}^2[\mu_{X|Y}^2 - 2\sigma_{X|Y}^2 v(y)] \\
        &- y^2 \mu_{X|Y}^2 + x^2 \sigma_{Y|X}^2 
        + 2(\sigma_{X|Y}^2 \mu_{Y|X} y - \sigma_{Y|X}^2 \mu_{X|Y} x) = 0
    \end{aligned}.
    \end{equation}

    Here we separated each term to a set of terms that only depend on $x$ and $y$ i.e to pairs of $f_k(x)g_k(x)$. Then we will use Theorem \ref{theorem:aczel} and define our functions as $\{\varphi_i(x)\} = \{\ \sigma_{y|x}^2, x\sigma_{y|x}^2, x^2\sigma_{y|x}^2 \}$ and $\{\Psi_i(y)\} = \{\ \sigma_{x|y}^2, y\sigma_{x|y}^2, y^2\sigma_{x|y}^2 \}$. By plugging in these values to equations \eqref{eq:aczel_1},\eqref{eq:aczel_2}, and \eqref{eq:aczel_3}, we obtain, 

    \begin{equation}
    \label{eq:equ_fx}
    \begin{pmatrix}
        2\sigma_{Y|X}^2 u(x) - \mu_{Y|X}^2 \\
        \sigma_{Y|X}^2 \\
        1 \\
        x^2\sigma_{Y|X}^2 \\
        \mu_{y|x} \\
        x\sigma_{Y|X}^2
        \end{pmatrix}
        =\underset{=\bm{B}}{
        \begin{pmatrix}
        A & B & C  \\
        1 & 0 & 0 \\
        D & E & F  \\
        0 & 0 & 1 \\
        G & H & J \\
        0 & 1 & 0 
        \end{pmatrix}}
        \begin{pmatrix}
         \sigma_{y|x}^2 \\
         x\sigma_{y|x}^2 \\
        x^2\sigma_{y|x}^2
        \end{pmatrix},
    \end{equation}
    \begin{equation}
    \label{eq:equ_gx}
   \begin{pmatrix}
       \sigma_{X|Y}^2 \\
        -2\sigma_{X|Y}^2 v(x) + \mu_{X|Y}^2 \\
        -y^2\sigma_{X|Y}^2 \\
        1 \\
        2y\sigma_{X|Y}^2 \\
        -2\mu_{X|Y}
        \end{pmatrix}
        =\underset{=\bm{C}}{
        \begin{pmatrix}
        1 & 0 & 0  \\
        K & L & M \\
        0 & 0 & -1  \\
        N & P & Q \\
        0 & 2 & 0 \\
        R & S & T 
        \end{pmatrix}}
        \begin{pmatrix}
         \sigma_{x|y}^2 \\
         y\sigma_{x|y}^2 \\
        y^2\sigma_{x|y}^2
        \end{pmatrix},
        \end{equation}where $\bm{B}\bm{C}'=\bm{0}$ and terms $A,B,\dots ,S,T$ are constants. These constants satisfy the following equations:
        \begin{equation}
        \begin{aligned}
            A &= -K, \quad L &= -2G, \quad M &= D, \\
            B &= -R, \quad S &= -2H, \quad E &= T, \\
            C &= -N, \quad P &= -2J, \quad Q &= F.
        \end{aligned}
        \end{equation}
        
        This is a direct consequence of $\bm{B}\bm{C}' = \bm{0}$ and we will them to obtain the following identities. 
        \begin{align}
            \label{eq:iden_fx_1}
            \sigma_{Y|X}^2 &= \frac{1}{D+Ex+Fx^2} \\
            \notag
            &\quad  \text{, obtained from the 3rd row of \eqref{eq:equ_fx}} \\
            \label{eq:iden_fy_1}
            \sigma_{X|Y}^2 &= \frac{1}{-C-2Jy+Fy^2}\\
            \notag
            &\quad   \\
            \notag
            &\quad  \text{, obtained from the 4th row of \eqref{eq:equ_gx}} \\
            \label{eq:iden_fx_2}
            \mu_{Y|X} &= \frac{G+Hx+Jx^2}{D+Ex+Fx^2}  \\
            \notag
            &\quad   \text{, obtained from the 5th row of \eqref{eq:equ_fx} and}  \eqref{eq:iden_fx_1} \\
            \label{eq:iden_fx_3}
            u(x) &= \frac{1}{2} \left[ A + Bx + Cx^2 + \frac{\mu_{Y|X}^2}{\sigma_{Y|X}^2} \right]   \\
            \notag
            &\quad  \text{, by direct substitution of \eqref{eq:iden_fx_1} to \eqref{eq:iden_fx_2}}
        \end{align}

        From \eqref{eq:u_x_v_y} we can write $f_{X,Y}(x,y)$ as, 
        \begin{align}
            \notag
            f_{X,Y}(x,y) &= \frac{1}{\sqrt{2\pi}} \exp\left\{u(x)-\frac{(y-\mu_{Y|X})^2}{2\sigma^2_{Y|X}}\right\} \\
            \notag
            &= \frac{1}{\sqrt{2\pi}} \exp\left\{\frac{1}{2} \left[ A + Bx + Cx^2 + \frac{\mu_{Y|X}^2}{\sigma_{Y|X}^2} \right] 
            \right. \\
            \notag
               &\qquad \left. - \frac{(y-\mu_{Y|X})^2}{2\sigma^2_{Y|X}}\right\}\\
            \notag
            &= \frac{1}{\sqrt{2\pi}} \exp\left\{\frac{1}{2} \Big[  A + Bx + Cx^2   \right. \\
            \notag
               &\qquad \left. - y^2(D+Ex+Fx^2)+2y(G+Hx+Jx^2)\Big]\right\}\\
            \label{eq:f_joing_step1}
            &= \frac{1}{\sqrt{2\pi}} \exp\left\{ 
                \begin{bmatrix}
                1 \\ x \\ x^2
                \end{bmatrix}^T
                \begin{bmatrix}
                A/2 & G & -D/2 \\
                B/2 & H & -E/2 \\
                C/2 & J & -F/2
                \end{bmatrix}
                \begin{bmatrix}
                1 \\
                y \\
                y^2
                \end{bmatrix}
                \right\},
        \end{align}
        where we used identities from \eqref{eq:iden_fx_1}, \eqref{eq:iden_fx_2}, and \eqref{eq:iden_fx_3}. With out loss of generality we can assume $\sigma_{Y|X}^2$ is a constant and does not depend on $x$. Then we can say that $E=F=0$ by looking at \eqref{eq:iden_fx_1}. We also know $\sigma_{X|Y}^2 > 0$ for any $y$, which implies $C + 2Jy < 0, \forall y$ according to \eqref{eq:iden_fy_1}. Thus $J$ must be zero and $C<0$. Similarly, $\sigma_{Y|X}^2 > 0 \implies D > 0$. Replacing these to \eqref{eq:f_joing_step1} we get,
        \begin{align}
            \label{eq:f_joing_step2}
            f_{X,Y}(x,y) &= \frac{1}{\sqrt{2\pi}} \exp\left\{ 
                \begin{bmatrix}
                1 \\ x \\ x^2
                \end{bmatrix}^T
                \begin{bmatrix}
                A/2 & G & -D/2 \\
                B/2 & H & 0 \\
                C/2 & 0 & 0
                \end{bmatrix}
                \begin{bmatrix}
                1 \\
                y \\
                y^2
                \end{bmatrix}
                \right\}.
        \end{align}

        This is equivalent to our formulation in Lemma \ref{lemma:bivariate_g}. Hence, we can say that $f_{X,Y}(x,y)$ is a bivariate density with normal marginals. 
        
\end{proof}

\subsection{Derivation for conditional score}
\label{sec:cond_score_est}

In section \ref{subsec:prior_approx}, we have seen that many of the assumptions in the existing conditional approximation methods have an implicit assumption that $p(\bm{x}_0)$ is Gaussian. Without loss of generality, we assume $p(\bm{x}_0) \sim \mathcal{N}(0, \sigma_{\bm{x}(0)}^2 \bm{I})$. We treat the prior variance $ \sigma_{\bm{x}(0)}^2 $ as a hyper-parameter to be tuned. This assumption is applied only to find a tractable and simple model for $p(\bm{y}|\bm{x}_t)$; the complex data prior is already learned by our pre-trained diffusion model $\bm{s}_{\theta}$. Using this assumption, we move to approximating the likelihood $ p(\bm{y}|\bm{x}_t)$. First, we will introduce an intermediate variable $\bm{z}$ defined as $\bm{z} = \mathcal{A}(\bm{x}_0)$. It is easy to show the following identities hold:
\begin{equation*}
 p(\bm{y}|\bm{z}) = \mathcal{N}(\bm{z}, \sigma_n^2 \bm{I})  \text{ and } p(\bm{z}|\bm{x}_t) = \mathcal{N}(\bm{A}\hat{\bm{x}}_0, \bm{A}\bm{\Sigma}_{\bm{x}(0|t)}\bm{A}^T).
 \end{equation*}

Using these identities, we can write the likelihood function as,

\begin{align}
    \notag
    p(\bm{y}|\bm{x}_t) &= \int  p(\bm{y}|\bm{z}, \bm{x}_t)p(\bm{z}|\bm{x}_t) d\bm{z} = \int p(\bm{y}|\bm{z})p(\bm{z}|\bm{x}_t) d\bm{z}  \\
    \notag 
    % &= \int  \underbrace{\frac{1}{C_1} \exp\left(-\bm{d}_{\bm{\Sigma}(1)}(\bm{y}-\bm{z})\right)}_{\mathcal{F}(\bm{x})}
    % \underbrace{\frac{1}{C_2} \exp\left(-\bm{d}_{\bm{\Sigma}_2}(\bm{z}-\mu_{\bm{z}})\right)}_{\mathcal{G}(\bm{x})} d\bm{z}  \\
    &= \int d\bm{z}  \underbrace{\frac{1}{C_1} \exp\left(-\frac{ \|\bm{y} - \bm{z}\|^2_2)}{2\sigma_n^2}\right)}_{\mathcal{F}(\bm{x})} \\&\quad \quad
    \underbrace{\frac{1}{C_2} \exp\left(-(\bm{z} - \mu_{\bm{z}|\bm{x}(t)})^T\bm{\Sigma}_{\bm{z}|\bm{x}(t)}^{-1} (\bm{z} - \mu_{\bm{z}|\bm{x}(t)})\right)}_{\mathcal{G}(\bm{x})}   \\
    &= \int \mathcal{F}(\bm{y}-\bm{z})\mathcal{G}(\bm{z}) d\bm{z} = \mathcal{F} \circledast \mathcal{G}
\label{eq:factorize_yxt}
\end{align}
where $C_1=\sqrt{(2\pi)^d \sigma_n^2}$ and $C_2 =\sqrt{(2\pi)^d |\bm{\Sigma}_{\bm{z}|\bm{x}(t)}|}$ are the scalar normalization constants for the Gaussian PDFs. Using the fact that the convolution of two Gaussian PDFs results in a Gaussian PDF (although not normalized) whose mean is the sum of the individual means and whose variance is the sum of the individual variances, we get
\begin{align}
\label{eq:adx_our_pyx_est}
p(\bm{y}|\bm{x}_t) &\approx \mathcal{N}(\bm{A}\hat{\bm{x}}_{0}, \sigma_n^2 \bm{I}  + \bm{A}\bm{\Sigma}_{\bm{x}(0|t)}\bm{A}^T).
\end{align}

\subsection{Efficient covariance inversion}
\label{sec:adx_cov_inv}
\textbf{Inpainting.}
For in painting, we can write the forward model as 

\[
\bm{y} = \mathcal{A}(x) + \eta = \bm{M} \odot \bm{X} + \eta,
\]

where $\bm{M}$ is the 2d-inpainting mask applied across all channels and $A = \text{diag}(\bm{M})$. Since $\mathcal{A}$ is a diagonal matrix with only ones and zeros, $\mathcal{A}\mathcal{A}^T = \mathcal{A} = \text{diag}(\bm{M})$. The inverse covariance matrix can then be written as

\[
\bm{\Sigma}_{t} = \frac{1}{\sigma_n^2} \bm{I} + \frac{1}{\sigma^2_{0|t}}\text{diag}(\bm{M}).
\]

The log likelihood (the terms that depend on $\bm{x}_t$) can then be computed as
\[
(\Delta \bm{y})^T 
 \bm{\Sigma}_{t}^{-1}  (\Delta \bm{y}) = \left\| \frac{ \bm{Y} - \bm{M}\odot \hat{\bm{x}}_0}{\sigma_n^2 \mathbf{1} + \sigma^2_{0|t}\mathbf{M} } \right\|^2_2.
\]
Note the division is element-wise and $\mathbf{1}$ is a matrix of all ones. In cases where we are unable to computed $\mathcal{A}\mathcal{A}^T$ effectively we will approximate it as the identity matrix. The conditional score can then be computed as
\begin{equation*}
 \bm{\kappa}_t = \gamma \bm{M} \odot \left(\frac{ \bm{y} - \bm{M}\odot \hat{\bm{x}}_0}{\sigma_n^2 \bm{1} + \sigma^2_{0|t}\bm{M} }\right).
\end{equation*}

\noindent\textbf{Deblurring.} For deblurring we can write the forward model as 

\[
 \bm{y} = \bm{h} \circledast \bm{x} + \bm{\eta} = \bm{H}\bm{x} + \bm{\eta}.
\]
We assume the blurring matrix $\bm{H}$ is the matrix representation of the cyclic convolution $ \bm{h} \circledast \bm{x}$. Under this assumption, we can represent $H$ as a doubly blocked block circulant matrix (BCCB). 

\begin{lemma}[\cite{olson2014circulant}]
    \label{lem:circ_diag}
    Let $\bm{M} \in \mathbb{R}^{n \times n}$ be a doubly blocked block circulant matrix (BCCB). Then $\bm{M}$ is diagonalizable, with its eigenvectors being the columns of the two-dimensional Discrete Fourier Transform matrix $\bm{F}$. Thus, we can factorize it as
    \[
    \bm{M} = \bm{F}^{H} (\Lambda) \bm{F},
    \]
    where $\Lambda = \bm{F}h$.

    % \begin{proof} Let $M_k$ be the circulant blocks that form the BCCM matrix $M$ and let $p_n(i)$
    %     \begin{align}
    %          M = \sum_{i=1}^{K} p_n(i) \otimes M_k   
    %     \end{align}
            
    % \end{proof}
\end{lemma}

Using Lemma \ref{lem:circ_diag}, the covariance matrix for deblurring is then written as 
\begin{align*}
    \bm{\Sigma}_{t} &= \sigma_n^2 \bm{I} + \sigma^2_{0|t} (\bm{F}^{H} \Lambda \bm{F})(\bm{F}^{H} \Lambda \bm{F})^{H} \\
    &= \sigma_n^2 \bm{I} + \sigma^2_{0|t} \bm{F}^{H} \Lambda \Lambda^{H} \bm{F}.
\end{align*}

Subsequently, the inverse covariance can be computed as

% using the Woodbury matrix identity \cite{hager1989updating} as 

\begin{align}
    \bm{\Sigma}_{t} ^{-1}
    &= (\sigma_n^2 \bm{F}^{H}\bm{F} + \sigma^2_{0|t} \bm{F}^{H} \Lambda \Lambda^{H} \bm{F})^{-1} \notag\\
    &= \left[\bm{F}^{H} \left(\sigma_n^2\bm{I} + \sigma^2_{0|t} \bm{\Lambda} \bm{\Lambda}^{H} \right) \bm{F} \right]^{-1} \notag\\
    &= \bm{F}^{H} \left(\sigma_n^2\bm{I} + \sigma^2_{0|t} \bm{\Lambda} \bm{\Lambda}^{H} \right)^{-1} \bm{F} \label{eq:apdx_inv_cov_deblur}
\end{align}

This inversion only requires inverting a diagonal matrix. The log likelihood (the terms that depend on $\bm{x}_t$) can then be computed as

\begin{align}
(\Delta \bm{y})^T 
 \bm{\Sigma}_{t}^{-1}  (\Delta \bm{y}) &= (\bm{F}\Delta \bm{y})^{H} \left(\sigma_n^2\bm{I} + \sigma^2_{0|t} \bm{\Lambda} \bm{\Lambda}^{H} \right)^{-1} \bm{F}\Delta \bm{y}
\end{align}

The conditional score can then be computed as (after substitution of \eqref{eq:apdx_inv_cov_deblur} to \eqref{eq:our_approx} )
\begin{equation*}
 \bm{\kappa}_t = \gamma \bm{F}^\mathsf{H}\bm{\Lambda}^\mathsf{H}\left(\frac{1}{\sigma_n^2}\bm{I} + \frac{1}{\sigma^2_{0|t}|\bm{\Lambda}|^2} \bm{I}\right)(\bm{F}\bm{y} - \bm{\Lambda} \bm{F}\bm{\hat{x}_0}).
\end{equation*}

\noindent\textbf{Super-resolution.} For super-resolution, we can write the forward model as 
\[
 \bm{y} = (\bm{h} \circledast \bm{x})_{\downarrow_d} + \bm{\eta} = \bm{S}\bm{H}\bm{x} + \bm{n},
\]
where $\bm{S}$ is an $m^2 \times n^2$ the down-sampling matrix with a factor $d$, $\bm{H}$ is the cyclic-convolution matrix. Using these identities, we can compute the inverse covariance matrix. 
\begin{align*}
    \bm{\Sigma_}{t} &= \sigma_n^2 \bm{I} + \sigma^2_{0|t} (\bm{S}\bm{F}^{H} \Lambda \bm{F})(\bm{S}\bm{F}^{H} \Lambda \bm{F})^{H} \\
    &= \sigma_n^2 \bm{I} + \sigma^2_{0|t} \bm{S}\bm{F}^{H} \Lambda \bm{F} \bm{F}^{H} \Lambda^{H}\bm{F}\bm{S}^{H}  \\    
    &= \sigma_n^2 \bm{I} + \sigma^2_{0|t} \bm{S}\bm{F}^{H} \Lambda\Lambda^{H}\bm{F}\bm{S}^{H}.  \\
    &= \sigma_n^2 \bm{I} + \sigma^2_{0|t} \bm{S}\bm{F}^{H} \Lambda\Lambda^{H}\bm{F}\bm{S}^{H}.  \\
\end{align*}
The inverse covariance can be computed as

% using the Woodbury matrix identity \cite{hager1989updating} as 

\begin{align}
    \bm{\Sigma}_{t} ^{-1} &= \bm{F}^{H} \left(\sigma_n^2\bm{I} + \sigma^2_{0|t}\bm{F}\bm{S}\bm{F}^{H} |\bm{\Lambda}|^2\bm{F}\bm{S}^{H}\bm{F}^{H}\right)^{-1} \bm{F} 
    \label{eq:apdx_sr_cov_inv}
\end{align}
Now we will present how we can efficiently invert the covariance matrix. First we define our down sampling matrix $\bm{S}$.

\begin{definition}
The down-sampling operator $\bm{S} \in \mathbb{R}^{m^2 \times n^2}$ is a decimation matrix that samples the first pixel in every non-overlapping $d \times d$ block of our image. Mathematically, we write $\bm{S}$  as a Kronecker product of two sparse matrices $\underline{\bm{S}} \in \mathbb{R}^{m \times n}$ that sample every $d$-th column. The relation is give by $\bm{S} = \underline{\bm{S}} \otimes \underline{\bm{S}}$.
\end{definition}

\begin{lemma}
    \label{lem:sub_circ}
    Let $\bm{M} \in \mathbb{R}^{n \times n}$ be a circulant matrix, and let $\underline{\bm{S}} \in \mathbb{R}^{m \times n}$ be a down-sampling matrix with a factor $d$, where $n$ is divisible by $d$ and let $m = n/d$. Then the matrix $\underline{\bm{S}}\bm{M}\underline{\bm{S}}^H$ is also a circulant matrix.
    \begin{proof} Since $\bm{M}$ is a circulant matrix, for any $i, j \in \{0,\dots,n-1\}$, we have $\bm{M}_{i,j} = \bm{M}_{(i-j)\%n}$. Let $\bm{A} = \underline{\bm{S}}\bm{M}\underline{\bm{S}}^{H}$ and $\bm{A}_{i,j} = \bm{M}_{i*d,j*d} $, assuming the matrix $\bm{S}$ and it's adjoint $\bm{S}^H$ sample every $d$-th rows and columns of $\underline{\bm{S}}\bm{M}\underline{\bm{S}}^H$ respectively . For any $i,j \in \{0,\dots,m-1\}$, we have 
    \[
     \bm{A}_{i,j}= \bm{M}_{i * d, j * d} = \bm{M}_{d(i-j)\%n} = \bm{M}_{(i-j)\%m}.
    \]
    Thus $\underline{\bm{S}}\bm{M}\underline{\bm{S}}^H \in \mathbb{R}^{m \times m}$ is a circulant matrix.
    
    \end{proof}
\end{lemma}
\begin{lemma} 
    The matrix $\underline{\bm{S}}\bm{M}\underline{\bm{S}}^H$
        is diagonalizable with diagonal entries given by 
        \begin{equation*}
          \bm{\lambda}(i) = \frac{1}{d} \sum_{j=0}^{d-1} \bm{\Lambda}(i+j \cdot m),
        \end{equation*}
        where $\bm{\Lambda}$ are the eigenvalues of the matrix $M$.
    \begin{proof}
        From Lemma \ref{lem:sub_circ}, we know $\underline{\bm{S}}\bm{M}\underline{\bm{S}}^{H}$ is diagonlizable. Now we will show how to obtain the diagonal elements. Using the diagonalizablity propriety, we can write the matrix as,
        \begin{align}
            \underline{\bm{S}}\bm{M}\underline{\bm{S}}^{H} &= \underline{\bm{S}}\bm{F_n}^H\bm{\Lambda}\bm{F_n}\underline{\bm{S}}^{H} \\
            \bm{F}_m^H(\bm{\lambda}) \bm{F}_m &= \underline{\bm{S}}\bm{F_n}^H(\bm{\Lambda})\bm{F_n}\underline{\bm{S}}^{H} \\
           \label{eq:diag_lemma_last}
           \bm{\lambda} &= \bm{F}_m \underline{\bm{S}}\bm{F_n}^H\bm{(\Lambda)}\bm{F_n}\underline{\bm{S}}^{H}  \bm{F}_m^H.
        \end{align}
        Here $\bm{F}_n$ represents is an $n\times n$ DFT matrix and $\bm{F}_m$ is an $m\times m$ DFT matrix. Now, let's take a closer look at the matrix $\bm{F}_m \underline{\bm{S}}\bm{F_n}^H$. 
    \begin{align}
        \bm{F}_m \underline{\bm{S}}\bm{F_n}^H&= \frac{1}{\sqrt{mn}} \begin{bmatrix}
        1  & \cdots & 1 \\
        \vdots & \vdots & \vdots \\
        1 & \cdots & \bar{\omega}^{(m-1)^2}
        \end{bmatrix} 
        \begin{bmatrix} 
            1 & 0 & \dots & 0 \\
            0 & \dots &  & 0 \\
            0 & 0 &1 & 0 \\
            0 & 0 & 0 & 0 \\
            \end{bmatrix}
            \notag \\& \quad \cdot
            \begin{bmatrix}
            1 & 1 & 1 & \cdots & 1 \\
            1 & \omega^{-1} & \omega^{-2} & \cdots & \omega^{-(n-1)} \\
            1 & \omega^{-2} & \omega^{-4} & \cdots & \omega^{-2(n-1)} \\
            \vdots & \vdots & \vdots & \ddots & \vdots \\
            1 & \omega^{-(n-1)} & \omega^{-2(n-1)} & \cdots & \omega^{-(n-1)^2}
            \end{bmatrix}   \\
            &= \frac{1}{\sqrt{mn}} 
            \begin{bmatrix}
            1 & 1 & \cdots & 1 \\
            1 & \bar{\omega} & \cdots & \bar{\omega}^{m-1} \\
            1 & \bar{\omega}^2 & \cdots & \bar{\omega}^{2(m-1)} \\
            \vdots & \vdots &\ddots & \vdots \\
            1 & \bar{\omega}^{m-1} &  \cdots & \bar{\omega}^{(m-1)^2}
            \end{bmatrix} \notag \\& \quad \cdot  
        \begin{bmatrix}
            1 & 1 & 1 & \cdots & 1 \\
            1 & \omega^{-d} & \omega^{-2d} & \cdots & \omega^{-d(n-1)} \\
            1 & \omega^{-2d} & \omega^{-4d} & \cdots & \omega^{-2d(n-1)} \\
            \vdots & \vdots & \vdots & \ddots & \vdots \\
            1 & \omega^{-d(m-1)} & \omega^{-2d(m-1)} & \cdots & \omega^{-d(m-1)(n-1)}
        \end{bmatrix}  \\
            \label{eq:dft_periodic}
            &= \frac{1}{\sqrt{mn}} 
           \begin{bmatrix}
            m \bm{I}_{m\times m} & m \bm{I}_{m\times m} & \cdots & m \bm{I}_{m\times m} 
        \end{bmatrix}_{m \times n},
    \end{align}
    where $\bar{\omega} = e^{\frac{-2\pi i}{m}}$ and $\omega = e^{\frac{-2\pi i}{n}}$. Equation \eqref{eq:dft_periodic} comes from the following identity. We can compute the entry at row $k$ and column $l$ of the matrix as,
    \begin{align*}
     (\bm{F}_m \underline{\bm{S}}\bm{F_n}^H)_{k,l} &= \sum_{j=0}^{m-1} \bar\omega^{kj}{\omega}^{-jdl} = \sum_{j=0}^{m-1} e^{\frac{-2\pi i kj}{m}}  e^{\frac{2\pi i jdl}{n}} \\ 
     &= \sum_{j=0}^{m-1} e^{2\pi ij (\frac{l-k}{m})}
    \end{align*}
    When $k = l$, the exponent becomes zero, resulting in the terms summing up to $m$. If $l - k < m$ and $l \neq m$, we can apply the geometric sum formula as
    \begin{equation*}
        \sum_{j=0}^{m-1} e^{2\pi ij (\frac{l-k}{m})} = \dfrac{1 - e^{2\pi i(l-k})}{1 - e^{2\pi i(\frac{l-k}{m})}} = 0 .
    \end{equation*}
    
    If $l-k > m$, due to the periodicity of the $N$th roots of unity, the values will be repeating. Thus we obtain $d$ blocks $m \times m$ identity matrices stacked horizontally.
    
    We can use these results to compute the entries of the diagonal matrix in \eqref{eq:diag_lemma_last},
       \begin{align} 
            \notag
            \bm{\lambda}
            &= \frac{1}{mn}  \begin{bmatrix}
                m \bm{I}_{m\times m} & \cdots & m \bm{I}_{m\times m} 
                \end{bmatrix}_{m \times n}
                \bm{\Lambda}_{n \times n}
                 \notag \\& \quad \cdot
                \begin{bmatrix}
                m \bm{I}_{m\times m} \\ m \bm{I}_{m\times m} \\ \vdots \\ m \bm{I}_{m\times m} 
                \end{bmatrix}_{n \times m} \\
            \notag
            &= \frac{1}{mn}  \begin{bmatrix}
                m \bm{I}_{m\times m} & m \bm{I}_{m\times m} & \cdots & m \bm{I}_{m\times m} 
            \end{bmatrix}_{m \times n}
             \notag \\& \quad \cdot
            \begin{bmatrix}
                m \text{diag}(\Lambda_0, \dots \Lambda_{m-1}) \\ 
                m \text{diag}(\Lambda_m, \dots \Lambda_{2m-1}) \\ \vdots \\ 
                m \text{diag}(\Lambda_{n-m-1}, \dots \Lambda_{n-1})
            \end{bmatrix}_{n \times m} \\
        \notag
         &= \frac{m}{n}  \sum_{j=0}^{d-1} \text{diag}(\Lambda_{jm}, \dots \Lambda_{jm-1}) \\
         &= \frac{1}{d} \begin{bmatrix}
             \Lambda_0 + \Lambda_m + \dots \Lambda_{m(d-1)} \\
             \Lambda_1 + \Lambda_{m+1} + \dots \Lambda_{2m-1)} \\
             \vdots \\
             \Lambda_{m-1} + \Lambda_{2m-1} + \dots \Lambda_{md-1}
         \end{bmatrix}
        \end{align}
    \end{proof}
    \end{lemma}
Now we will use these lemmas to compute efficient inversion of \eqref{eq:apdx_sr_cov_inv}.  
\begin{align*}
    \bm{F}\bm{S}\bm{F}^{H} &= (\bm{F}_m \otimes \bm{F}_m)(\underline{\bm{S}} \otimes \underline{\bm{S}}) (\bm{F}_n \otimes \bm{F_n})^H \\ &= \frac{1}{d} 
           \begin{bmatrix}
            \bm{I}_{m\times m} & \cdots & \bm{I}_{m\times m} 
        \end{bmatrix}_{m \times n} \otimes \\ & \quad\quad\quad \begin{bmatrix}
            \bm{I}_{m\times m} & \cdots & \bm{I}_{m\times m} 
        \end{bmatrix}_{m \times n}.
\end{align*}

We observe that the matrix $\bm{F} \bm{S} \bm{F}^{H}$ consists of blocks that are either $m \times m$ identity matrices (scaled by $1/d$) or $m \times m$ zero matrices. When this matrix multiplies $\bm{\Lambda}^2$ from both the left and the right, the result is another diagonal matrix $\bm{\Gamma}$, where the entries are the averages of $d^2$ elements in $\bm{\Lambda}^2$. In particular, the diagonal elements of $\bm{\Gamma}$ are given by
% \small
\begin{align} 
    \notag
    \bm{\Gamma} 
    &= \frac{1}{d^2} \begin{bmatrix}
        \Lambda_0 + \dots + \Lambda_{m(d-1)} 
        \\ + \Lambda_{m^2d} + \dots + \Lambda_{m^2d + m(d-1)}
        + \pmb{\dots} \\
        + \Lambda_{m^2d(d-1)} + \dots + \Lambda_{(m^2d + m)(d-1)}
        \\[1.5em]
        % second row
        \Lambda_1 + \dots + \Lambda_{m(d-1)+1} 
        \\ + \Lambda_{m^2d + 1} + \dots + \Lambda_{m^2d + m(d-1) + 1} 
        + \pmb{\dots} \\
        + \Lambda_{m^2d(d-1) + 1} + \dots + \Lambda_{(m^2d + m)(d-1) + 1}
       \\[1.5em]
        % third row
        \pmb{\vdots}
        \\[1.5em]
        % last row
        \Lambda_{(md+1)(m-1)} + \dots + \Lambda_{(md+1)(m-1) + m(d-1)} \\
        + \Lambda_{(md+1)(m-1) + m^2d} + \pmb{\dots} \\
        + \Lambda_{md(md-1) + m -1} + \dots + \Lambda_{m^2d^2 - 1} 
    \end{bmatrix}
\end{align}
Equivalently, for $ i \in \{0,1,\dots m^2\}$, each element of $\Gamma$ can be written as,
\begin{equation*}
          \bm{\Gamma({i})} = \frac{1}{d^2} \sum_{j=0}^{d-1}\sum_{k=0}^{d-1} \bm{\Lambda}[m\left(jdm + k + d \lfloor i/m \rfloor \right) + \text{rem}(i,m)],
\end{equation*}
where $\text{rem}(i,m)$ is the remainder in the division $i/m$. Finally, we can write the conditional score as 
\begin{equation*}
     \bm{\kappa}_t = \gamma \bm{F}^\mathsf{H}\bm{\Lambda}^\mathsf{H}\bm{F}\bm{S}^\mathsf{H}\left({\frac{1}{\sigma_n^2} \bm{I} + \frac{1}{\sigma^2_{0|t}|\bm{\Gamma}|^2}\bm{I}}\right)(\bm{y} - \bm{S}\bm{H}\bm{\hat{x}_0}).
\end{equation*}

\noindent\textbf{Separable systems.} We can write the forward operator of a separable system as $\bm{A}\bm{x} = \bm{A}_l\bm{X}\bm{A}_r^T$, where $\bm{A}=\bm{A}_l \otimes \bm{A}_r$, $\bm{x} \in \mathbb{R}^{n^2,1}$, $\bm{X} \in \mathbb{R}^{n\times n}$, $\bm{A}_l, \bm{A}_r \in \mathbb{R}^{m \times n}$ and $\bm{A} \in \mathbb{A}^{m^2 \times n^2}$. Let us denote the SVD decomposition of these matrices as $\bm{A}=\bm{U}\bm{\Sigma}\bm{V}^T$,  $\bm{A}_l=\bm{U}_l\bm{\Sigma}_l\bm{V}_l^T$, and  $\bm{A}_r=\bm{U}_r\bm{\Sigma}_r\bm{V}^T_r$. We also have the following properties from Kronecker products: $\bm{A}\bm{x}=\bm{A}_l\bm{X}\bm{A}_R^T$, $\bm{U}=\bm{U}_l\otimes \bm{U}_r $, $\bm{\Sigma}=\bm{\Sigma}_l\otimes \bm{\Sigma}_r $ and , $\bm{V}=\bm{V}_l\otimes \bm{V}_r $. The conditional score is given by 
\begin{align*}
 \bm{\kappa}_t &\simeq  \bm{A}^T
    \left(\sigma_n^2 \bm{I}  + \bm{A}\bm{\Sigma}_{\bm{x}(0|t)}\bm{A}^T \right)^{-1} (\bm{y} - \bm{A}\hat{\bm{x}}_0) \\
    &= \bm{V}\bm{\Sigma}\bm{U}^T
    \left(\sigma_n^2 \bm{I}  + \bm{\sigma}_{\bm{x}(0|t)}\bm{U}\bm{\Sigma}^2\bm{U}^T \right)^{-1} (\Delta\bm{y})\\
    &= \bm{V}\bm{\Sigma}\bm{U}^T
    \left(\sigma_n^2 \bm{U}\bm{U}^T  + \bm{\sigma}_{\bm{x}(0|t)}\bm{U}\bm{\Sigma}^2\bm{U}^T \right)^{-1} (\Delta\bm{y})\\
    &= \bm{V}\bm{\Sigma}
    \left(\sigma_n^2 \bm{I} + \bm{\sigma}_{\bm{x}(0|t)}\bm{\Sigma}^2 \right)^{-1} \bm{U}^T(\Delta\bm{y})\\
    &= (\bm{V}_l\otimes \bm{V}_r)(\bm{\Sigma}_l\otimes \bm{\Sigma}_r)\\& \quad
    \left(\sigma_n^2 \bm{I} + \bm{\sigma}_{\bm{x}(0|t)}(\bm{\Sigma}_l\otimes \bm{\Sigma}_r)^2 \right)^{-1} (\bm{U}_l\otimes \bm{U}_r)^T(\Delta\bm{y})\\
    &= \bm{V}_l \bm{\Sigma}_l
    \left( \dfrac{\bm{U}_l^T\Delta\bm{Y}\bm{U}_r}{\sigma_n^2 \bm{1}\bm{1}^T + \bm{\sigma}_{\bm{x}(0|t)}(\bm{\sigma}_l \bm{\sigma}_r^T)^2} \right)\bm{\Sigma}_r \bm{U}_r^T,
\end{align*}
where the last step is obtained by repeatedly applying $(\bm{A}_l \otimes \bm{A}_r)\bm{x}=\bm{A}_L\bm{X}\bm{A}_R^T$.

\section{Experimental details}
\label{sec:imp_details}
In this section, we outline the implementation and hyper-parameter details of our proposed method as well as the competing methods.

\subsection{Experimental setup}
We utilized existing pre-trained diffusion models for both datasets, for FFHQ we obtained it from \cite{chung2023diffusion} and for ImageNet from \cite{dhariwal2021diffusion}. Note that both of these network are trained as generative denoisers for their respective datasets and are not fine-tuned for any image restoration task. We use these models for all methods in our comparison set that use diffusion models. Following \cite{chung2023diffusion}, we use the first 1K images of the FFHQ dataset as validation data. We resize the original $1024 \times 1024$ images in this dataset to $256 \times 256$. For ImageNet, we obtain the preprocessed $256 \times 256$ images from \cite{dhariwal2021diffusion} and use the first 1K images as our validation dataset. We use these validation datasets for all comparison experiments. For performance metrics, we report both standard reconstruction metrics, including peak signal-to-noise ratio (PSNR) and structural similarity index (SSIM) \cite{wang2004image}, as well as perceptual metrics, including Fréchet inception distance (FID) \cite{heusel2017gans} and Learned Perceptual Image Patch Similarity (LPIPS) \cite{zhang2018unreasonable}.

\subsection{Inverse problems. } 
\label{subsec:apdx_inv_prob_detais}
The first inverse problem we consider is image super resolution. The task here is to recover an image blurred using a $9 \times 9$ Gaussian blur kernel with standard deviation $3.0$ followed by a $\times 4$ down-sampling using a decimation matrix. 
The down-sampling operator is a decimation matrix that samples the first pixel in every non-overlapping $4 \times 4$ block of our image. 
For Gaussian deblurring, we use a $61 \times 61$ kernel with standard deviation $3.0$ and for motion delubrring we randomly generate a kernel\footnote{{https://github.com/LeviBorodenko/motionblur}} with size $61 \times 61$ with intensity of $0.5$. Finally, we consider random image inpainting, where we remove pixels randomly in all color channels. Each pixel can be removed with a uniform probability in the range $[0.7, 0.8]$. For noisy experiments, we apply an additive Gaussian noise with standard deviation $\bm{\sigma}_n=0.05$.

\subsection{Implementation details}
\textbf{CoDPS.} Our algorithm has two main parameters the prior data covariance $\sigma_{x(0)}^2$ and the gradient scalar term $\{\zeta_t\}$ as shown in Algorithms  \ref{alg:co_dps_ddim} and \ref{alg:co_dps}. For the DDIM sampling, we set $\eta=1$ and NFE$=100$ for all experiments. We have to main hyper-parameters that we tune for various problems. These are $\{\zeta_t\}_{t=1}^{t=T}$ and $\sigma^2_{\bm{x(0)}}$. For $\{\zeta_t\}_{t=1}^{T}$, we start with a fixed value and then reduce it once based on a condition related to $t$.  We report the values used in the main paper for each inverse problems in Tables \ref{tab:hyperprams_imgnet_ffhq} and \ref{tab:hyperprams_noiseless_imgnet}. We also report a variant of our method that uses DDPM sampling in Algorithm \ref{alg:co_dps}.  Both DDIM and DDPM algorithms requires the total number of diffusion steps $N$, a measurement vector $\bm{y}$, scaling parameters $\{\zeta_t\}_{t=1}^N$,  noise standard deviation ${\{\tilde\sigma_t\}_{t=1}^N}$, and prior variance hyper-parameter $\sigma_0^2$ and outputs the final reconstructed image $\hat{\bm{x}}_0$. 

\begin{table*}[htp]
\centering
\caption{Hyperparameters used for CoDPS for noisy inverse problems on FFHQ and ImageNet datasets}
\label{tab:hyperprams_imgnet_ffhq}
\resizebox{0.85\textwidth}{!}{%
\begin{tabular}{ccccccc}
\hline
         &      &           & $\times 4$ SR & Gaussian Deblur & Motion Deblur & Inpainting \\ \hline
    FFHQ     & DDIM & $\zeta_t$   &
   $\begin{cases} 
    5\text{e-2} & \text{if } t > 15 \\
    1\text{e-2} & \text{otherwise }
    \end{cases}$  & 
    $\begin{cases} 
    5\text{e-2} & \text{if } t > 15 \\
    5\text{e-3} & \text{otherwise }
    \end{cases}$        & 
    $\begin{cases} 
    5\text{e-2} & \text{if } t > 15 \\
    2.5\text{e-3} & \text{otherwise }
    \end{cases}$   & 
    $\begin{cases} 
    4\text{e-1} & \text{if } t > 10 \\
    4\text{e-3} & \text{otherwise }
    \end{cases}$ 
\\  \cdashline{2-7} 
         &      & $1/\sigma^2_{\bm{x(0)}}$ & 2     & 1               & 1             & 0          \\
 \hdashline
         & DDPM & $\zeta_t$   & 
        $\begin{cases} 
    1\text{e-2} & \text{if } t > 100 \\
    1\text{e-3} & \text{otherwise }
    \end{cases}$  & 
    $\begin{cases} 
     1\text{e-2} & \text{if } t > 200 \\
    1\text{e-3} & \text{otherwise }
    \end{cases}$        & 
    $\begin{cases} 
     1\text{e-2} & \text{if } t > 200 \\
    1\text{e-3} & \text{otherwise }
    \end{cases}$   & 
    $\begin{cases} 
    4\text{e-1} & \text{if } t > 80 \\
    4\text{e-4} & \text{otherwise }
    \end{cases}$  
\\  \cdashline{2-7} 
         &      & $1/\sigma_{\bm{x}(0)}^{2}$  & 0.1   & 0.1             & 0.1           & 5e2        \\ \hline
    ImageNet & DDIM & $\zeta_t$   &  
    $\begin{cases} 
        5\text{e-2} & \text{if } t > 20 \\
        1\text{e-2} & \text{otherwise }
        \end{cases}$  & 
        $\begin{cases} 
        4\text{e-2} & \text{if } t > 15 \\
        8\text{e-3} & \text{otherwise }
        \end{cases}$        & 
        $\begin{cases} 
         4\text{e-2} & \text{if } t > 25 \\
         1\text{e-2} & \text{otherwise }
        \end{cases}$   & 
         -  
         \\  \cdashline{2-7} 
         &      & $1/\sigma_{\bm{x}(0)}^{2}$  & $2$     & $2$            & $2$ & -      
    \\
          \hdashline
         & DDPM & $\zeta_t$   &  
         $\begin{cases} 
    2\text{e-2} & \text{if } t > 200 \\
    2\text{e-3} & \text{otherwise }
    \end{cases}$  & 
    $\begin{cases} 
    1\text{e-2} & \text{if } t > 200 \\
    1\text{e-3} & \text{otherwise }
    \end{cases}$        & 
    $\begin{cases} 
    5\text{e-3} & \text{if } t > 100 \\
    1\text{e-4} & \text{otherwise }
    \end{cases}$   & 
    $\begin{cases} 
    2\text{e-2} & \text{if } t > 80 \\
    2\text{e-4} & \text{otherwise }
    \end{cases}$    \\\cdashline{2-7} 
         &      & $1/\sigma_{\bm{x}(0)}^{2}$  & $1e-2$     & $1e-1$            & $1$            & $1\text{e1}$       
\\ \hline
\end{tabular}%
}
\end{table*}

\begin{table*}[htp]
\centering
\caption{Hyperparameters used for CoDPS for noiseless inverse problems on ImageNet dataset}
\label{tab:hyperprams_noiseless_imgnet}
\resizebox{0.75\textwidth}{!}{%
\begin{tabular}{cccccc}

\hline
         &      &           & $\times 4$ SR & Gaussian Deblur & Motion Deblur  \\ \hline
    ImageNet & DDIM & $\zeta_t$   &  
    $\begin{cases} 
        6\text{e-2} & \text{if } t > 10 \\
        4\text{e-3} & \text{otherwise }
        \end{cases}$  & 
        $\begin{cases} 
        1\text{e-1}  & \text{if } t > 10 \\
        2\text{e-3}  & \text{otherwise }
        \end{cases}$        & 
        $\begin{cases} 
        1\text{e-1}  & \text{if } t > 10 \\
        2\text{e-3}  & \text{otherwise }
        \end{cases}$         \\ \cdashline{2-6} 
         &      & $1/\sigma_{\bm{x}(0)}^{2}$  & $1\text{e-1}$     & $1$            & $1\text{e-1}$          
    \\ \hline
    \end{tabular}%
    }
    \end{table*}

\begin{algorithm}[t]
    \caption{CoDPS(DDPM)}
   \label{alg:co_dps}
    \begin{algorithmic}[1]
     \REQUIRE $N$, $\bm{y}$, \{$\zeta_t\}_{t=1}^N,  {\{\tilde\sigma_t\}_{t=1}^N}$, $\sigma_{x_0}^2$
     \STATE $\bm{x}_N \sim \mathcal{N}(\bm{0}, \bm{I})$
      \FOR{$t=N-1$ {\bfseries to} $0$}
        
         \STATE{{$\hat{\bm{s}} \gets \bm{s}_\theta(\bm{x}_t, t)$}}
         
         \STATE{{$\hat{\bm{x}}_0 \gets 
         \dfrac{1}{\sqrt{\bar\alpha_t}}(\bm{x}_t + (1 - \bar\alpha_t)\hat{\bm{s}})$}}
         
         \STATE{$\bm{z} \sim \mathcal{N}(\bm{0}, \bm{I})$}
         
         \STATE{$\bar{\bm{x}}_{t} \gets \dfrac{\sqrt{\alpha_t}(1-\bar\alpha_{t-1})}{1 - \bar\alpha_t}\bm{x}_t + \dfrac{\sqrt{\bar\alpha_{t-1}}\beta_t}{1 - \bar\alpha_t}\hat{\bm{x}}_0 +  {\tilde\sigma_t \bm{z}}$}
         
         \STATE{ {
         {$\Delta{\bm{y}} \gets \bm{y} - \mathcal{A} \left(\dfrac{\sigma_{x_0}^{2}\sqrt{\bar\alpha_t}\bar{\bm{x}}_{t} + (1-\bar\alpha_t)\hat{\bm{x}}_0}{(1-\bar\alpha_t) + \sigma_{x_0}^{2}\bar\alpha_t}\right) $ }}}
         
         \STATE{ 
         {$\bm{x}_{t-1} \gets \bar{\bm{x}}_{t} -  {\zeta_t}\nabla_{\bar{\bm{x}}_{t}} (\Delta \bm{y})^T (\sigma_n^2 \bm{I} + \sigma^2_{0|t}\mathcal{A}\mathcal{A}^T)^{-1}  (\Delta \bm{y})$}}
      \ENDFOR
      \STATE {\bfseries return} $\hat{\bm{x}}_0$
    \end{algorithmic}
\end{algorithm}

\textbf{DPS. } We utilized the official source code \cite{chung2023diffusion} and used the exact hyperparameters outlined in the original paper for each problem. We used DDPM sampling and set NFE to 1000 for all experiments.

\textbf{$\Pi$GDM. } We obtained the implementation of $\Pi$GDM from \cite{song2022pseudoinverse} for the noiseless inverse problems. We used the pseudo-code from the paper to implement the noisy version of the algorithm. We used DDIM sampling and set NFE to 100 for all experiments.

\textbf{DMPS.} We used the publicly available implementation of DMPS \cite{meng2022diffusion}. For super-resolution and Gaussian deblurring on the FFHQ dataset, we followed the configurations provided in the source code. Since the codebase did not include an image inpainting operator or configuration, we implemented it and used a default scale parameter of $1.75$.

\textbf{DDRM. } We utilized the official implementation of DDRM \cite{kawar2022denoising}. For all experiments, we used the default setting reported in the paper of $\mu_B = 1.0$ and $\eta=0.85$. We used DDIM sampling and set NFE to 20 for all experiments.

\textbf{DiffPIR. } We used the publicly available implementation of DiffPIR \cite{zhu2023denoising}. We used the exact hyperparameters reported in the paper, with NFE set to 100 for all experiments.

{\textbf{DDS.} We used the publicly available version of DDS~\cite{chung2024decomposed}. We included our forward operator implementations, since the original code only supported medical imaging tasks. We tuned the regularizer parameter $\gamma$ between range $0.5-2$ using grid search to obtain the best result for each problem.}

{\textbf{DDNM.} We used the official implementation of DDNM~\cite{wangzero}. For the noiseless deblurring and super-resolution experiments, we used the SVD versions. We set $\eta=0.85$ for all experiments.}

{\textbf{ReSample.} We used the publicly available implementation of ReSample~\cite{songsolving}. For each task, we used the corresponding configuration files. We modified the noise-level configuration to $\sigma_y=0.05$. For each problem, we tuned the scale parameters within the range $0.1$–$1$.}

\textbf{PnP-ADMM.} We used the Scientific Computational Imaging Code (SCICO) \cite{balke2022scientific} library to implement PnP-ADMM with a pre-trained 17-layer DnCNN \cite{zhang2017beyond} network. The number of maximum iterations was set to $50$ for all experiments. We performed a grid search to determine the optimal ADMM penalty parameter $\rho$, exploring values between $0.01$ and $0.4$ per experiment.

\textbf{ADMM-TV.} Similar to PnP-ADMM, we used SCICO \cite{balke2022scientific} to address our inverse problems with ADMM-TV. The number of maximum iterations was set to $50$ for all experiments. A grid search was conducted to identify the best $\rho$ parameter, ranging from $0.01$ to $0.4$, and the TV-norm regularizer $\lambda$, ranging from $0.001$ to $0.01$.

\section{Additional results and figures}
\label{sec:additional}
\subsection{Experiment on mixture of Gaussians.}

\begin{figure*}[t]
  \centering
  \includegraphics[width=\textwidth]{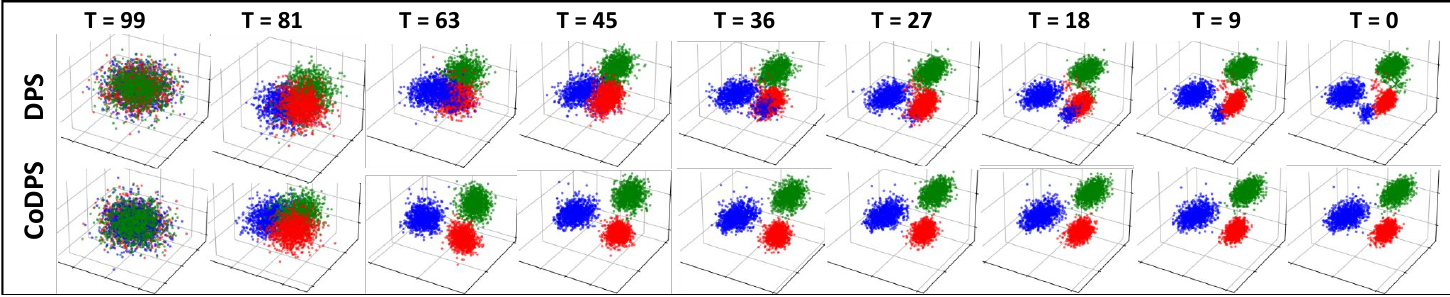}
  \caption{Reconstruction by diffusion posterior sampling using DPS (top) and CoDPS (ours). We show our method reconstructs measurements more accurately and the reconstructed sample distribution converges to the true posterior data distribution. (Data distribution and measurements for this example is shown in Figure \ref{fig:data_meas_GMM})}
  \label{fig:intro_mix_gauss_recon}
\end{figure*}

\label{sebsec:adx_exp_on_mix_gaus}
As a proof of concept, we designed a simple task where we solve an ill-posed inverse reconstruction problem from projected measurements to a lower-dimensional space. This is a generalization of the sample shown in figure \ref{fig:intro_mix_gauss_recon}. We define our prior distribution as a Gaussian mixture model, $  p(\bm{x}) = \sum_{k=1}^{K} \pi_k \mathcal{N}(\bm{x} \mid \bm{\mu}_k, \bm{\Sigma}_k),$ where $\bm{x} \in \mathbb{R}^n$, $\pi_k$ are the mixture weights that satisfy $\sum_{k=1}^{K} \pi_k = 1$ and $\pi_k \geq 0$. We will use a diffusion network to learn this prior distribution and perform posterior sampling using our proposed method. 

\textbf{Learning the data prior.} In this experiment, we limit the number of mixtures to $k=3$. We build our dataset by generating samples from an $n$-dimensional GMM model, where $n$ is a randomly selected integer between $3$ and $15$. We randomly set the means of mixtures using random $n$ dimension vectors selected uniformly form the range $[-3,3]$ and the covariance matrices using an $n\times n$ symmetric positive definite matrix whose entries are sampled uniformly from the range $[0,1]$. Now that we have a data at hand the first task it to learn its distribution using a diffusion model. We use an MLP-based model similar to the one used for toy problems in \cite{sohl2015deep}. Our diffusion model consists of 4 MLP blocks, each with an input and output feature size of 64. Additionally, the model has input and output layers to map features between the data dimension $n$ and 64. 
  
The time scale $t$ is concatenated to the input of the model and the number of diffusion steps is set to $T=100$. We trained the network for $500$ epochs and verify that we can sample from the correct distribution using DDPM reverse sampling \cite{ho2020denoising}. 

\textbf{Generating measurements.} To obtain measurements $\bm{y}$,we first generate a random projection matrix $\bm{A} \in \mathbb{R}^{m\times n}$ where $m<n$ to obtain measurements using the forward model \eqref{eq:fwd_model}, where the noise level is set as $\bm{\sigma}_n=0.05$. For a given $n$, we set $m \in \{2,\dots, n-1\}$ randomly. An example of these measurements along with the true data distribution is shown in figure \ref{fig:data_meas_GMM}. 

 \textbf{Posterior sampling.} To simplify our sampling we consider the conditioned distribution of $\bm{x}$ on a one-hot encoded latent  $\bm{z} \in \mathbb{R}^K$, where $p(\bm{z}_k=1) = \pi_k$. We can then write the conditional density as $p(\bm{x}|\bm{z}_k=1) =  \mathcal{N}(\bm{x} \mid \bm{\mu}_k, \bm{\Sigma}_k)$. We aim to obtain the MAP estimate $p(\bm{x|y,\bm{z}_k})$ by performing posterior sampling using using DPS \cite{chung2023diffusion}, our method CoDPS, and a modified version of our method where we assume $\bm{A}\bm{A}^T=\bm{I}$. The modified (simplified) CoDPS is important because in cases where we can not invert and apply the covariance correction term efficiently we resort to simplifying the expression as $\bm{A}\bm{A}^T = \bm{I}$. For these three methods, we compare the conditional covariance approximation error, that is the error between the true posterior covariance and estimated covariance from generated samples. The MAP estimate, $\hat{\bm{x}}$, of $\bm{x}$ conditioned on $\bm{y}$ and $\bm{z}_k$, which is also the Minimum mean square estimate (MMSE), has a covariance given by,
 \begin{equation}
        \label{eq:adx_map_est_gmm_cond}
     {\bm{\Sigma}}_{\hat{\bm{x}}|\bm{y},\bm{z}(k)} = \left(\bm{\Sigma}_k^{-1} + \frac{1}{\sigma_n^2} \bm{A}^T\bm{A} \right)^{-1}
 \end{equation}
 Figure \ref{fig:cov_approx_err} shows the approximate error of ${\bm{\Sigma}}_{\hat{\bm{x}}|\bm{y},\bm{z}(k)}$ and the sample covariance of the reconstructions using the posteriors sampling sampling methods DPS and CoDPS. The idea is that the sample covariance of the reconstructed estimations for each cluster in the GMM should  approximate and converge to \eqref{eq:adx_map_est_gmm_cond}. We observe that both our methods yield better approximation that improves as we obtain more measurements. The simplified CoDPS ($\bm{A}\bm{A}^T=\bm{I}$) performs much better than DPS and comparable to the version of our method that uses the correct $\bm{A}\bm{A}^T$. As the method accounts for the uncertainty in estimates $\hat{\bm{x}}_0$ \eqref{eq:var_x0_xt_approx}, we expect an improved result over DPS as discussed in \ref{subsec:prior_approx}. 

Figure \ref{fig:data_meas_GMM} shows the source data and noisy projected measurements for the Gaussian mixture experiment reported in section \ref{sebsec:exp_on_mix_gaus}, and the time evolution shown in Figure \ref{fig:intro_mix_gauss_recon}.

\begin{figure}[ht]
    \centering
    \includegraphics[width=\columnwidth]{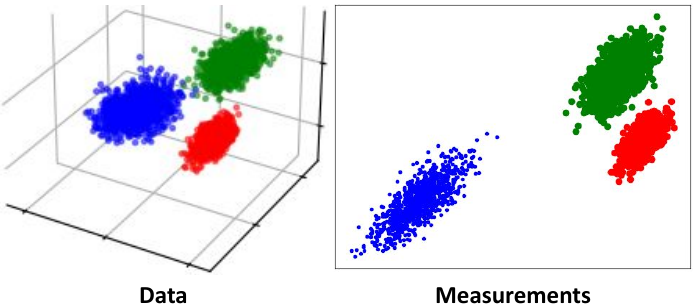}
     \caption{Source data and noisy measurements for Gaussian Mixture Model experiments. }    \label{fig:data_meas_GMM}
\end{figure}

\subsection{Additional experiments and results}

\subsubsection{Hyperparameter ablation study}
To assess the sensitivity of our proposed methods to hyperparameters, we conducted an ablation study based on the reviewer’s suggestion. Specifically, we used the $\sigma_0^2$ and gradient scaling schedule ${\zeta_t}$ tuned for the super-resolution task (SR $\times4$), and then applied the same hyperparameters directly to the Gaussian deblurring, motion deblurring, and inpainting tasks on the FFHQ dataset.

The results in Table \ref{tab:hyperparam_ablation} show that CoDPS maintains strong performance on both Gaussian and motion deblurring tasks. We even observed some improvement in terms of PSNR, despite a slight performance drop in other metrics. This shows our method is not highly sensitive across tasks. Only the inpainting task showed a degradation when using shared parameters. This can be due to its fundamentally different degradation process.

\begin{table}[ht]
\centering
\caption{Ablation study using hyperparameters tuned for SR ($\times 4$) and applied directly to other tasks. Results show that CoDPS remains competitive without per-task tuning on two new  tasks.}
\label{tab:hyperparam_ablation}
\resizebox{\textwidth}{!}{
\begin{tabular}{llllllllllllll}
\toprule

{} & \multicolumn{4}{c}{\textbf{Gaussian Debluring}} & \multicolumn{4}{c}{\textbf{Motion Deblurring}} & \multicolumn{4}{c}{\textbf{Inpainting}} \\

\cmidrule(lr){2-5}
\cmidrule(lr){6-9}
\cmidrule(lr){10-13}

{\textbf{Method}} & 
{PSNR $\uparrow$} & {SSIM $\uparrow$} & {FID $\downarrow$} & {LPIPS $\downarrow$} & 
{PSNR $\uparrow$} & {SSIM $\uparrow$} & {FID $\downarrow$} & {LPIPS $\downarrow$} & 
{PSNR $\uparrow$} & {SSIM $\uparrow$} & {FID $\downarrow$} & {LPIPS $\downarrow$} \\

\midrule

CoDPS \\(using SR4x params) & \textbf{27.78}
  & 0.794  & 36.75   &  0.258 
  & \textbf{27.22} & 0.725  &  44.38 &  0.301
  & 17.47 & 0.419  & 155.92 & 0.514 \\

\hdashline
CoDPS \\(finetuned per task) & 27.56
  & \textbf{0.799}  & \textbf{33.13}  &  \textbf{0.239} & 27.15
  &  \textbf{0.771} & \textbf{26.00}  & \textbf{0.243}  & \textbf{27.32}
  & \textbf{0.795}  & \textbf{37.87} &  \textbf{0.236}\\

\bottomrule
\end{tabular}
}
\end{table}

\label{subsec:apdx_add_figs}
In Table \ref{tab:results_linear_imgnet_inpaint} and \ref{tab:results_linear_imgnet_psnrssimlpips}, we report results on the ImageNet dataset. Our method achieved the best SSIM score while showing competitive performance to DPS in terms of PSNR and LPIPS scores.

\begin{table}[t]
\vspace{-4em}
\centering
\caption{
Performance metrics of inverse problems on \textbf{ImageNet} validation dataset. \textbf{Bold}: best, \underline{underline}: second best.
}
\resizebox{\textwidth}{!}{
\begin{tabular}{lllllllllllllll}
\toprule
{} && \multicolumn{4}{c}{\textbf{SR ($\times 4$)}} 
&
\multicolumn{4}{c}{\textbf{Deblur (Gaussian)}} & \multicolumn{4}{c}{\textbf{Deblur (motion)}}\\

\cmidrule(lr){3-6}
\cmidrule(lr){7-10}
\cmidrule(lr){11-14}

% \cmidrule(lr){10-11}
{\textbf{Method}} & {\textbf{NFE}} & 
{PSNR $\uparrow$} & {SSIM $\uparrow$} & {FID $\downarrow$} & {LPIPS $\downarrow$} & {PSNR $\uparrow$} & {SSIM $\uparrow$} & {FID $\downarrow$} & {LPIPS $\downarrow$}  & 
{PSNR $\uparrow$} & {SSIM $\uparrow$} & {FID $\downarrow$} & {LPIPS $\downarrow$} \\

\midrule

\thead[l]{$\Pi$GDM \cite{song2022pseudoinverse}} % SR
& 100 &  21.23 & 0.482 & 88.93  & 0.495   
% Gaus Deblur
&  18.29 &0.341  & 135.86  & 0.610   
% Motion Deblur
& 23.89  & 0.663 &  52.15 & 0.373   \\

\thead[l]{DPS \cite{chung2023diffusion}} % SR
& 1000 & 21.70 & 0.562& \textbf{45.20} & \textbf{0.381}  
% Gaus Deblur
&  21.97 & \textbf{0.706}& 62.72 & 0.444
% Motion Deblur
& 20.55  & 0.634 & 56.08 & 0.389 \\

\thead[l]{DiffPIR \cite{zhu2023denoising}}
& 100 & 20.58 & 0.430 & 110.69 & 0.560  
% Gaus Deblur
& 21.62  & 0.504 & 58.61 & 0.433
% Motion Deblur
& \textbf{24.61}  &  \textbf{0.661} & \underline{46.72}  &  \underline{0.359}\\

\thead[l]{DDRM~\cite{kawar2022denoising}} 
& 20 &\textbf{23.64}  & \textbf{0.624}& 67.69 & 0.423  
% Gaus Deblur
& 22.73 & \underline{0.705}& 63.02 & 0.427
% Motion Deblur
&  - &  - & -& -\\

\thead[l]{MCG~\cite{chung2022improving}} 
& 1000 &13.39  & 0.227 & 144.5 &  0.637
% Gaus Deblur
& 16.32  & 0.441 & 95.04 & 0.550
% Motion Deblur
&  5.89 & 0.037  & 186.9 & 0.758\\

\thead[l]{PnP-ADMM~\cite{chan2016plug}}
&  - & 20.43 & 0.538 & 159.58 & 0.592  
% Gaus Deblur
&  20.77 & 0.545 & 105.96 & 0.502
% Motion Deblur
&  21.02 & 0.563 & 108.53 & 0.500\\

\thead[l]{ADMM-TV~\cite{goldstein2009split}} 
% SR
& - & 18.40  & 0.490 & 191.97 & 0.600
% Gaus Deblur
& 21.52  & 0.594 & 108.10 & 0.495
% Motion Deblur
&  20.39 & 0.557 & 148.80 & 0.524\\

\cmidrule(l){1-14}
\textbf{\thead[l]{CoDPS }} 
% SR
& 1000 &\underline{23.12}  & \underline{0.620} & \underline{54.62} &   \underline{0.397} 
% Gaus Deblur
&  \textbf{24.07} & 0.693  & \textbf{47.70} & \textbf{0.348}
% Motion Deblur
& \underline{23.36}  & \underline{0.653} & \textbf{37.48} &  \textbf{0.333}\\

\textbf{\thead[l]{CoDPS }} 
% SR
& 100 & 22.17  & 0.551 & 64.78 &  0.429 
% Gaus Deblur
&  \underline{23.74} &  0.640  & \underline{55.03} & \underline{0.374}
% Motion Deblur
&   22.90 & 0.552  &  68.60 & 0.424 \\
\bottomrule
\end{tabular}
}
\label{tab:results_linear_imgnet_psnrssimlpips}
\end{table}
% -------------------------------------------------------------

\begin{table}[t]

% \vspace{-1em}
\centering
\caption{
Performance metrics of inverse problems on \textbf{ImageNet} 256$\times$256-1k validation dataset. 
}
\resizebox{0.5\textwidth}{!}{
\begin{tabular}{lllll}

\toprule
{} 
& \multicolumn{4}{c}{\textbf{Inpainting (random)}}
\\

\cmidrule(lr){2-5}

{\textbf{Method}} & {PSNR $\uparrow$} & {SSIM $\uparrow$} & {FID $\downarrow$} & {LPIPS $\downarrow$} 
 \\

\midrule

\thead[l]{DPS \cite{chung2023diffusion}} 
% Inp random
& \textbf{25.02} & \underline{0.750} & \underline{32.21} & 0.269 \\

\thead[l]{DiffPIR \cite{zhu2023denoising}}
% Inp random
&  - & - & - & -\\

\thead[l]{DDRM~\citep{kawar2022denoising}} 
% Inp random
& 22.06  & 0.593 & 84.44 & 0.437\\

\thead[l]{MCG~\citep{chung2022improving}} 
% Inp random
& 22.62& 0.608 & \textbf{30.51} & 0.316 
\\

\thead[l]{PnP-ADMM~\citep{chan2016plug}}
% Inp random
&  21.41 & 0.641 & 65.55 & 0.403  \\ 

\thead[l]{ADMM-TV} 
% Inp random
& 18.94 & 0.546 & 138.40 & 0.540  \\

\cmidrule(lr){1-5}

\textbf{\thead[l]{CoDPS (Ours)}} 
% Inp random
& \underline{24.87}& \textbf{0.761} & 39.67 & \underline{0.276} \\

\bottomrule
\end{tabular}
}
\label{tab:results_linear_imgnet_inpaint}
\end{table}

We present additional results related to the experiments reported in our main paper in Figures \ref{fig:supp_fig_sr}, \ref{fig:supp_fig_mbr}, \ref{fig:supp_fig_gdbr}, \ref{fig:supp_fig_sr4x_imgnet}, and \ref{fig:add_imgnet_noiseless_gdbr}. 

\begin{figure*}[ht]
    \centering
    \includegraphics[width=\textwidth]{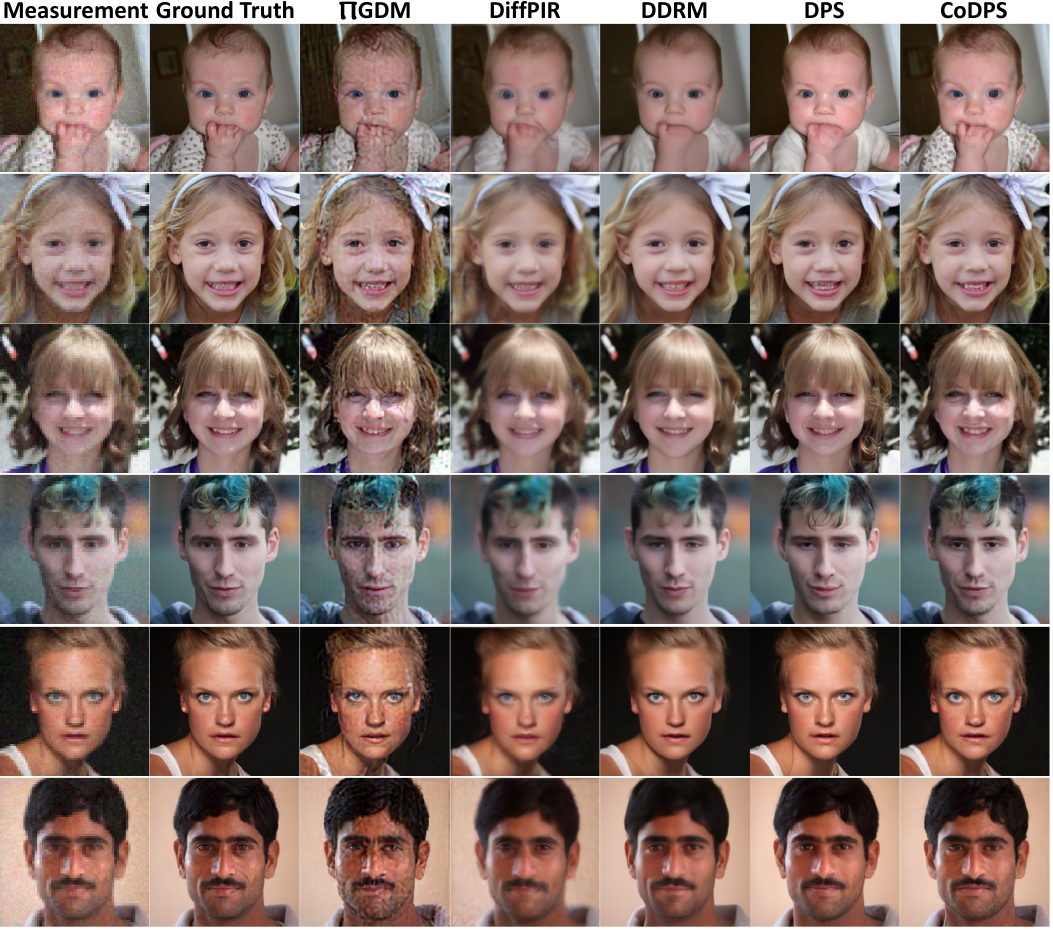}
    \caption{Motion deblurring results on FFHQ dataset}
    \label{fig:supp_fig_sr}
\end{figure*}

\begin{figure*}[ht]
    \centering
    \includegraphics[width=\textwidth]{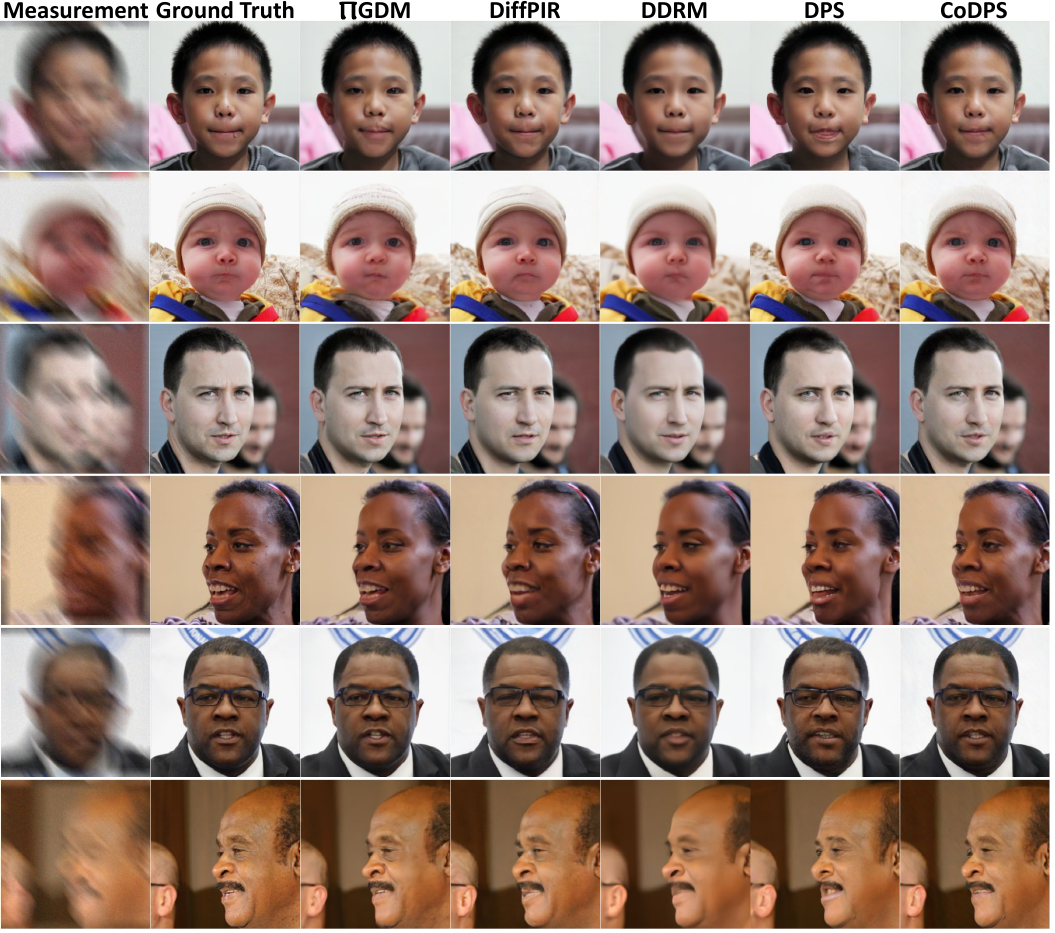}
    \caption{Super-resolution results on FFHQ dataset}
    \label{fig:supp_fig_mbr}
\end{figure*}

\begin{figure*}[ht]
    \centering
    \includegraphics[width=\textwidth]{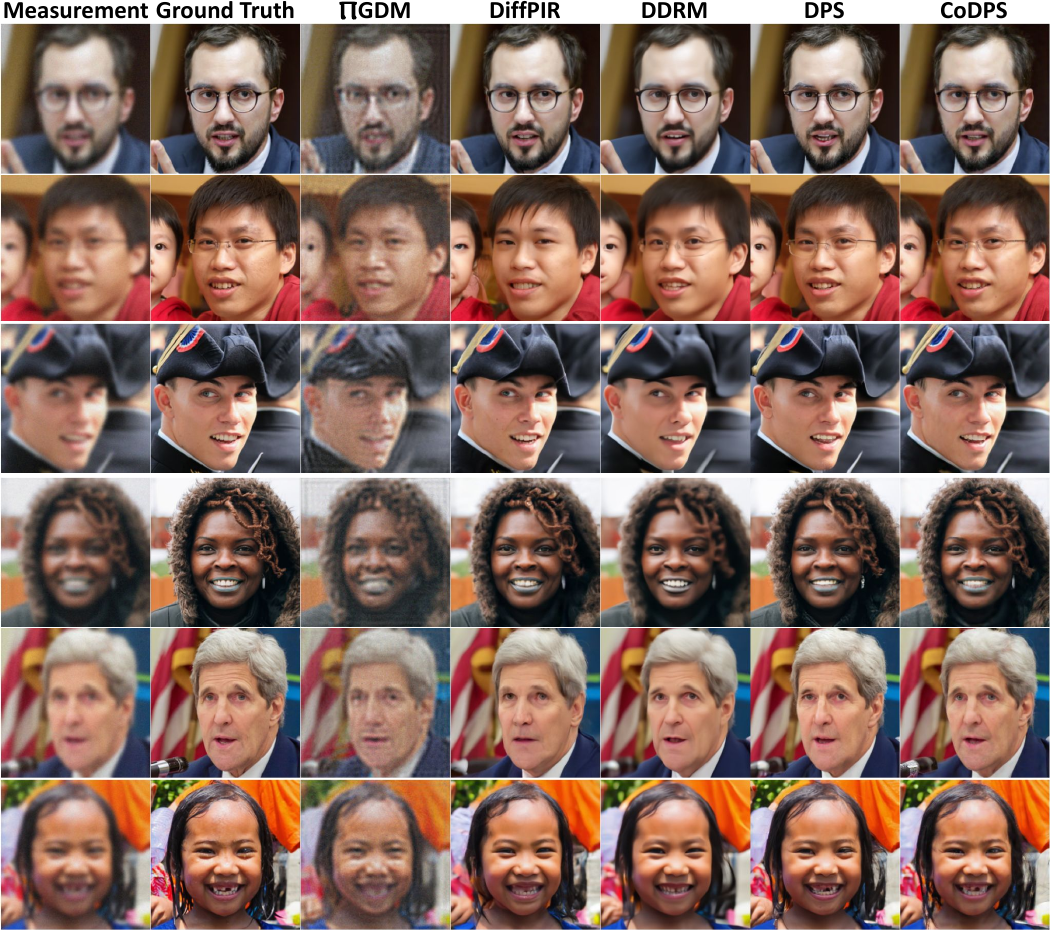}
    \caption{Gaussian deblurring results on FFHQ dataset}
    \label{fig:supp_fig_gdbr}
\end{figure*}

\begin{figure*}[ht]
    \centering
    \includegraphics[width=\textwidth]{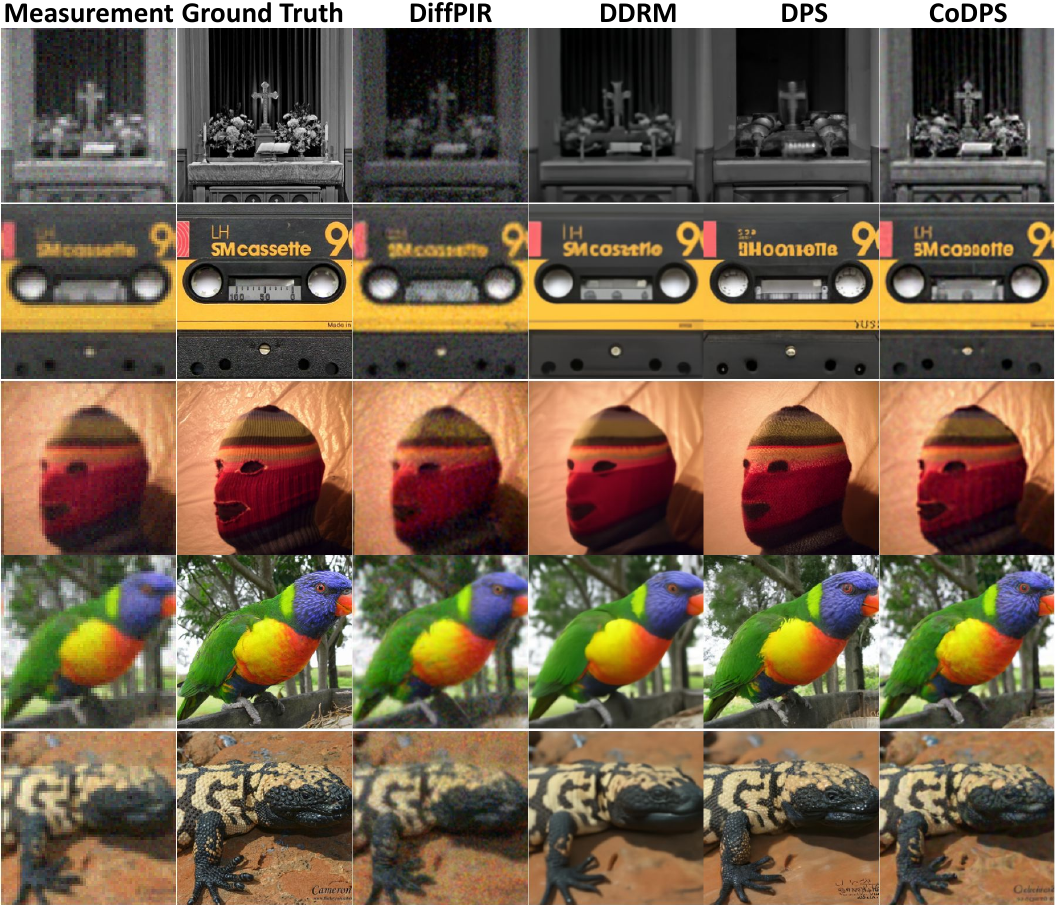}
    \caption{Noisy super-resolution $\times$4 experiments on the ImageNet dataset}
    \label{fig:supp_fig_sr4x_imgnet}
\end{figure*}

\begin{figure*}
    \centering
    \includegraphics[width=0.7\textwidth]{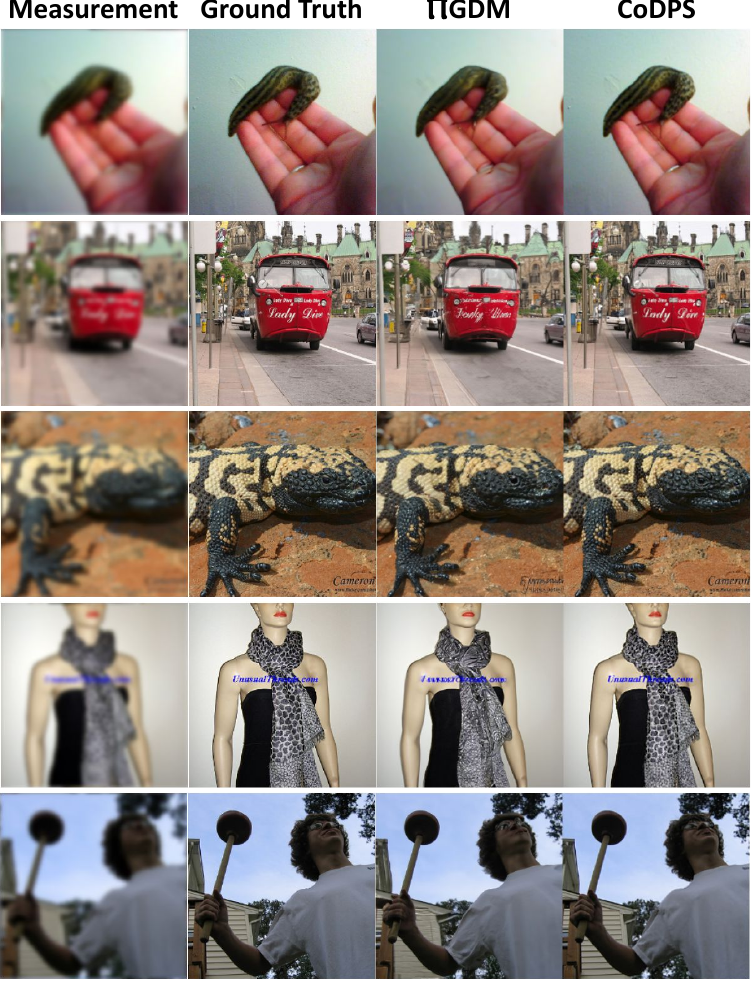}
    \caption{Noiseless Gaussian deblurring experiments on ImageNet dataset}
    \label{fig:add_imgnet_noiseless_gdbr}
\end{figure*}

\end{document}